\documentclass{elsarticle}
\usepackage{fullpage}

\usepackage{amssymb}               
\usepackage{amsmath}               
\usepackage{amsthm}                
\usepackage{color}
\usepackage{enumitem}
\usepackage{comment}
\usepackage{marginnote}

\usepackage{hyperref}

\newcommand{\Real}{\ensuremath{\mathbb R}}
\newcommand{\mylabel}[1]{\label{#1}}
\newcommand{\ind}[1]{\hspace*{#1em}}

\newcommand{\If}{{\bf if\ }}
\newcommand{\Then}{{\bf then\ }}
\newcommand{\Else}{{\bf else\ }}

\newcommand{\Return}{{\bf return\ }}

\newcommand{\diag}{{\rm diag}}

\newcommand{\length}{{\rm bitlength}}
\newcommand{\Z}{\ensuremath{\mathbb Z}}
\newcommand{\Q}{\ensuremath{\mathbb Q\mskip1mu}}
\newcommand{\Rel}{\ensuremath{\mathcal R}}
\newcommand{\K}{{\mathsf{K}}}

\newcommand{\B}{{\mathsf{B}}}
\newcommand{\M}{{\mathsf{M}}}

\newcommand{\boldz}{{\mathbf 0}}

\newcommand{\Znn}{\ensuremath{\mathbb Z}^{n\times n}}

\newtheorem{theorem}{Theorem}

\newtheorem{definition}[theorem]{Definition}
\newtheorem{corollary}[theorem]{Corollary}

\newtheorem{remark}[theorem]{Remark}
\newtheorem{lemma}[theorem]{Lemma}
\newtheorem{example}[theorem]{Example}

\DeclareMathOperator{\rowmod}{{\mathbf r}mod}
\DeclareMathOperator{\colmod}{{\mathbf c}mod}

\DeclareMathOperator{\loglog}{loglog}

\newlength{\algwidth}
\usepackage{etoolbox}
\makeatletter
\patchcmd{\ps@pprintTitle}
  {Preprint submitted to}
  {Using style file from }
  {}{}
\makeatother

\definecolor{orange}{rgb}{1.0, 0.5, 0.0} 

\begin{document}
\title{Computing bases in Hermite normal form of lattices of integer relations}

\author{George Labahn}
\address{Cheriton School of Computer Science, University of Waterloo,
Waterloo ON, Canada N2L 3G1}
\ead{glabahn@uwaterloo.ca}

\author{Arne Storjohann}
\address{Cheriton School of Computer Science, University of Waterloo,
Waterloo ON, Canada N2L 3G1}
\ead{astorjoh@uwaterloo.ca}

\begin{abstract}
Given a full column rank $M \in \Z^{\ell \times m}$ and
an $F \in \Z^{n \times m}$ 
we present an algorithm  to compute the $n \times n$
basis in Hermite form of the integer lattice comprised of all rows
$p \in \Z^{1 \times n}$ such that $pF \in \Z^{1 \times m}$ is in
the integer lattice generated by the rows of $M$.  The algorithm is randomized
of the Las Vegas type, that is, it can fail with probability at most
$1/2$, but if fail is not returned it guarantees to produce the 
correct result.  When $M$ is
square and $F=I_m$, then the computed basis is the Hermite normal
form of $M$, and the algorithm uses about the same number of bit
operations as required to multiply together two matrices of the
same dimension and size of entries as $M$.
\end{abstract}


\maketitle


\section{Introduction}

There has been considerable recent work on designing algorithms for
computations on integer (and polynomial) matrices that have about
the same cost as multiplying together two matrices of the same
dimension and size of entries as the input matrix.  For a nonsingular
$m \times m$ integer input matrix $M$ the target complexity is
$(m^\omega \log ||M||)^{1+o(1)}$ bit operations, where $\omega$ is
the exponent of matrix multiplication, $||M|| = \max_{ij} |M_{ij}|$
denotes the largest entry in absolute value, and the ``$+o(1)$''
in the exponent indicates a missing factor $c_1(\log n)^{c_2}
(\loglog ||M||)^{c_3}$ for positive real constants $c_1,c_2,c_3$.
This complexity has been achieved with a deterministic algorithm
for computing the rational solution to a linear system $\vec{x}
M=\vec{b}$~\cite{BirmpilisLabahnStorjohann19}, and with Las Vegas
randomized algorithms for computing the diagonal Smith normal form
$S$ of $M$~\cite{BirmpilisLabahnStorjohann20} and the Smith form
with unimodular multipliers $MV=US$~\cite{BirmpilisLabahnStorjohann21}.

However, no previous algorithm  with the desired complexity exists
for computing the Hermite normal form of an integer matrix.
Recall that a full column rank $\ell \times m$ matrix
\begin{equation}\label{intro:hermite}
 \left [ \begin{array}{cccc}
h_1 & h_{12} & \cdots & h_{1m} \\
& h_2 & \cdots & h_{2m} \\
& & \ddots & \vdots \\
& & & h_m \\
& & & 
\end{array} \right ] \nonumber
\end{equation}
is in (row) Hermite form if all its entries are nonnegative, the
off-diagonal entries $h_{\ast j}$ are strictly smaller than the
diagonal entry $h_j$ in the same column and any zero rows are at
the bottom.  For every full column rank matrix $M \in \Z^{\ell
\times m}$, there is a unimodular matrix $W \in \Z^{\ell \times
\ell}$ such that $H = W M$ is in Hermite form.  The form is unique
with its existence dating back to 1851~\cite{Hermite}.  The first
$m$ rows of $H$, which we call the Hermite basis of $M$, gives a
canonical presentation for ${\mathcal L}(M)$, the lattice generated
by the $\Z$-linear combinations of the rows of $M$.

\subsection{Main contributions}

In this paper we present a new algorithm for computing the Hermite
form of a full column rank input matrix $M\in\Z^{\ell \times m}$
using $(\ell m^{\omega-1} \log ||M||)^{1+o(1)}$ bit operations.
The algorithm is Las Vegas randomized, that is, it can report
\textsc{fail} with probability at most $1/2$, and if \textsc{fail}
is not reported,  then the output is guaranteed to be correct.  More
precisely we have:

\begin{theorem} \mylabel{thm:ione}
Let $M \in \Z^{\ell \times m}$ have full column rank. There exists
a Las Vegas randomized algorithm that computes the Hermite form of
$M$ in $O(\ell m^{\omega-1}\,\B(D/m + \log m)(\log m)^2))$ bit
operations, where $D$ is a bound for the sum of the bitlengths of
the columns of $M$.
\end{theorem}

Here, $\B(t)$ bounds the cost of operations on integers with bitlength
bounded by~$t$, including gcd-related operations, and by ``bitlength
of a column'' we mean the bitlength of the largest entry in absolute
value.  The term $D/m$ in the cost estimate is thus a bound for the
average column bitlength of $M$.  We always have $D/m \in
O(\log ||M||)$.  Section~\ref{ssec:cost} gives our cost model.

Our approach is to solve the  more general problem of finding the
{\em  Hermite bases of an integer relations lattice}.  Given $M \in
\Z^{\ell  \times  m}$ having full column rank and $F \in \Z^{n
\times m}$, the set \begin{equation} \mylabel{eq:introrel} \Rel(M,F)
:= \left  \{ p \in \Z^{1 \times n} \mid pF = \boldz \bmod M \right
\} \end{equation} defines an integer relations lattice for $F$ with
modulus $M$.  Here the notation $A = \boldz \bmod M$ stands for
``$A= Q M$ for some integer matrix $Q$,''  that is, the rows of $A$
are in the lattice generated by the rows of $M$.  A basis for the
lattice $\Rel(M,F)$ is any nonsingular matrix $P \in \Z^{n \times
n}$ that has minimal determinant in absolute value among all matrices
that satisfy $PF = \boldz \bmod M$.  In this paper we present an
algorithm for computing the unique basis for $\Rel(M,F)$ that is
in Hermite form.

\begin{theorem} \mylabel{thm:itwo}
Let $M \in \Z^{\ell \times m}$ have full column rank and $F \in
\Z^{n \times m}$.  There exists a Las Vegas randomized algorithm
that computes the Hermite basis $H$ of $\Rel(M,F)$ in
$$ O((\ell+n)m^{\omega-1}\, \B(D/m + \log m) (\log m)^2) $$
plus 
$$ O(n^{\omega}\, \B(D/n  + \log n)(\log n)^2) $$
bit operations, where  $D$ is a bound for the sum of the bitlengths
of the columns of $M$ and $F$.
\end{theorem}
As an example consider when $F=I_m$ in~(\ref{eq:introrel}).  Then
$\Rel(M,I)$ is equal to ${\mathcal L}(M)$ and thus the problem of
computing the Hermite basis of $\Rel(M,I)$ is precisely that solved
by Theorem~\ref{thm:ione}.  As a second example, if $M \in \Z^{m \times m}$ 
is a diagonal matrix of arbitrary nonzero moduli,
then the  lattice of solutions to the linear diophantine system
$\vec{x} A = \vec{b} \bmod M$ is given by the integer relations lattice
$$\Rel \left (M ,\left[ \begin{array}{c}
- \vec b \\ A\end{array} \right]\right ).$$ 
The Hermite basis of this lattice gives a compact presentation of
the general solution to the system.  Intersections of
integer lattices can also be be conveniently described by integer
relations lattices.

\subsection{Historical background}

The last few decades has seen considerable progress in reducing the
cost of computations on polynomial matrices over $\K[x]$.  Over
$\K[x]$ the complexity counts the number of required field operations
from $\K$, and instead of $\log ||M||$ for integer matrices the
measure of size is a bound on the degrees of entries in the input
matrix. For polynomial matrices  there are deterministic reductions
to matrix multiplication for the problems linear
solving~\cite{GuptaSarkarStorjohannValeriote11},
nullspace/kernel~\cite{ZhouLabahnStorjohann2012}, column
bases~\cite{ZhouLabahn13}, minimal or order basis~\cite{ZhouLabahn09},
minimal interpolation bases~\cite{JeannerodNeigerSchostVillard17},
determinant and Hermite normal forms~\cite{LabahnNeigerZhou17}.  In
the case of polynomial Smith forms there is a Las Vegas randomized
algorithm~\cite{Storjohann03a}.

Effective computation of Hermite normal forms of integer and
polynomial matrices has a long history. Procedures developed during
the 1970--1990's produced algorithms that were first provably
polynomial time~\cite{KannanBachem}. In the case of integer matrices
a series of papers~\cite{ChouCollins,DomichKannanTrotter, Iliopoulos89:1,
HafnerMcCurley}  reduced the  complexity to  $(n^{4} \log
||M||)^{1+o(1)}$ bit operations.  Matrix multiplication was
incorporated later~\cite{StorjohannLabahnISSAC96,thesis} to get
algorithms having a worst case complexity $(n^{\omega+1} \log
||M||)^{1+o(1)}$.  These algorithms, although presented for integer
matrices, could also be modified to work for polynomial matrices
having similar complexities with size measured by degrees.  Later
a series of techniques based on heuristic methods, for example found
in~\cite{MicciancioWarinschi01, PernetStein10, PauderisStorjohann13,
LiuPan19}, were presented which  reduced this complexity in some
important cases. However, these heuristic algorithms also require
strong assumptions, for example, that there only be a small number
of non-trivial ($\neq 1$) late diagonal entries of the Hermite form.
Most recently, a Las Vegas randomized 
algorithm~\cite{BirmpilisLabahnStorjohann23b} was given having complexity
$(n^3 \log ||M||)^{1 + o(1)}$ for any nonsingular $M$.

Relation lattices $\Rel(M,F)$ have mostly previously appeared for
$M$ and $F$ being matrices of polynomials $\K[x]$. For example,
Pad\'e approximants and more generally Hermite-Pad\'e approximants
are included when $M = x^N$, while simultaneous Pad\'e approximants
are included when $M$ is a diagonal matrix of powers of $x$. Rational
interpolants are included when $M$ is a diagonal matrix of linear
factors describing  interpolation conditions. Finding bases for the
$\K[x]$ modules is useful as it describes all possible approximants
or interpolants for a given rational expression. This includes, for
example,  descriptions of the well-known block structures of Pad\'e
and multi-point Pad\'e approximation tables~\cite{BeckermannLabahn97}.
Fast algorithms for basis computation in these cases include those
in~\cite{BeckermannLabahn94} and~\cite{NeigerVu17}.

\subsection{Our approach and an outline of the paper} \mylabel{sec:approach} 

We view the problem of computing Hermite forms in terms of finding
{\em minimal denominators} of matrices of rational numbers.  If $H$
is the Hermite form of a nonsingular integer matrix $M$ then $HM^{-1}$
is unimodular over $\Z$.  It follows that $\det H$ is minimal in
absolute value among all nonsingular integer matrices that clear
the denominators of $M^{-1}$ upon premultiplication.  Examples and
some basic properties of minimal denominators can be found in
\cite[Section~3.1]{BirmpilisLabahnStorjohann23b}.

Instead of working with $M^{-1}$ explicitly we use a compact
approximation of the form $F S^{-1}$, where $F$ is an integer matrix
and $S$ is the Smith form of $M$.  The pair $(S,F)$ is called a
{\em Smith massager} for $M$, and was introduced in
\cite[Definition~1]{BirmpilisLabahnStorjohann21} for efficiently
computing Smith forms.  Here we use the fact that $H$ is the Hermite
basis of $\Rel(S,F)$.  Section~\ref{ssec:sm} gives the background
and results on Smith  massagers needed in this paper.

A Smith massager $(S,F)$ has two important properties. First, $S$
and $F$ are right coprime as integer matrices, which is equivalent
to having $\det S = \det H$, where $H$ is the Hermite basis of
$\Rel(S,F)$.  Second, because $S$ is diagonal its modulus action
in the equation $H F = \boldz \bmod S$  is decoupled with respect
to the columns, which means that we can assume without loss of
generality that $(S,F)$ is a \emph{reduced} Smith massager, with
the entries in each column of $F$ having  magnitude strictly less
than the corresponding diagonal entries of $S$. These properties
ensure that the total size in bits required to write down $(S,F)$
is bounded in the worst case by that required to write down the
Hermite basis $H$ of $\Rel(S,F)$.

Our main result  is a recursive algorithm to compute the Hermite
basis corresponding to a Smith massager. During the course of the
algorithm integer relations lattices with a description $\Rel(M,G)$
can arise that are not necessarily in the form that we require.
In particular, $M$ may have more rows than columns, $M$ may not be
in Smith form, and $M$ and $G$ may not be coprime.  In order to
handle such an arbitrary input we show how to construct from $(M,G)$
a Smith massager $(S,F)$ for the Hermite basis of $\Rel(M,G)$.  The
construction of $(S,F)$ uses a sequence of transformations:
\begin{equation}\mylabel{eq:tf}
 (M,G)
\stackrel{1}{\rightarrow} \left (\left [ \begin{array}{c} M_1 \\\hline M_2 \end{array} \right ], G \right )
\stackrel{2}{\rightarrow} \left (\left [ \begin{array}{c} S_1 \\\hline M_3 \end{array} \right ], G_1\right)
\stackrel{3}{\rightarrow} (T_1, G_1)
\stackrel{4}{\rightarrow} (S_2, G_2) 
\stackrel{5}{\rightarrow} (K, C)
\stackrel{6}{\rightarrow} (S, F).
\end{equation}
Transformations 1--6 achieve, in turn, that $M_1$ is a nonsingular
submatrix of $M$, $S_1$ is the Smith form of $M_1$, $T_1$ is the
Hermite basis of ${\mathcal L}(S_1) + {\mathcal L}(M_3)$ (the lattice
generated by the rows of $S_1$ and $M_3$), $S_2$ is the Smith form
of $T_1$, $(K,C)$ has inputs which are coprime as integer matrices,
and $S$ is the Smith form of $K$.  Section~\ref{sec:relations}
explains in detail these transformations and establishes their
correctness, that is, that the transformed input generates the same
lattice of integer relations.

The transformations mentioned above 
imply that
computing a Hermite basis for $\Rel(M,G)$ reduces to finding a
Hermite basis for  $\Rel(S,F)$ where $(S,F)$ is a Smith massager.
The main idea of our recursive algorithm for computing the Hermite
basis of $\Rel(S,F)$ is to split  into two subproblems based on the
columns of the output basis $H$.  The algorithm uses a parameter
$m$ that is an upper bound on the number of nontrivial columns (with
diagonal entry $>1$) in the Hermite basis being computed.  If $H$
is $n \times n$ then at the top level of the recursion we have
$m=n$.  For any partition $m=m_1+m_2$ we can uniquely factor
$H=H_2H_1$ as
\begin{equation}\mylabel{eqn:subprobs}
H = H_2 H_1  = 
{\left [ \begin{matrix} ~ I_{m_1} & H_{12}~ \\
 & \bar{H}_2 ~\end{matrix} \right ]}
{
\left [ \begin{matrix} ~\bar{H}_1 & \\
 &  I_{m_2}  \end{matrix} \right ]} = \left [ \begin{matrix} ~ \bar{H}_1 & H_{12} \\
 & \bar{H}_2 \end{matrix} \right]. 
 \end{equation}
Combining these two factors into a single Hermite basis is free.
The difficulty in this case is defining the two subproblems.  In
Section~\ref{sec:alg} we show how to use a subset of the transformations
in~(\ref{eq:tf}) to efficiently construct Smith massagers for $H_1$
and $H_2$, resulting in an algorithm for computing $H.$

The computations providing the biggest challenge for finding Smith
massagers for each of the two subproblems of our algorithm are
transformations~3 and~5.
For transformation~3 note that the input matrix
\begin{equation} \mylabel{eq:toherm1}
  \left [ \begin{array}{c}  S_1 \\\hline M_3 \end{array} \right ] 
\mbox{~~has Hermite form~~}
 \left [ \begin{array}{c}  T_1 \\\hline \phantom{\boldz} \end{array} \right ],
\end{equation}
where $T_1$ is the Hermite basis of ${\mathcal L}(S_1) + {\mathcal
L}(M_3)$, as required.  It turns out that the presence of $S_1$
means that $T_1$ can be recovered by computing the row \emph{Howell
form} of this input matrix over $\Z/(s)$, with $s$ being the largest
invariant factor of $S_1$.  A fast algorithm exists for the Howell
form, and working modulo $s$ can already give an improved bit
complexity compared to using a ``modulo determinant'' Hermite form
algorithm, which for this input matrix could work modulo $\det S_1$.
For transformation~5 a similar approach also works. If $T_2$ is the
Hermite basis of ${\mathcal L}(S_2) + {\mathcal L}(G_2)$, then
\begin{equation} \mylabel{eq:toherm2}
  \left [ \begin{array}{ccc}  T_2 & & I  \\ G_2 & I &  \\
  S_2 & &  \end{array} \right ]
\mbox{~~has Hermite form with shape~~}
 \left [ \begin{array}{ccc}  T_2 & & \ast\\ 
 & I & C \\
 & & K  \end{array} \right ],
\end{equation} 
and $\Rel(S_2,G_2)=\Rel(K,C)$ with the latter having coprime inputs.
The block $T_2$ which is needed to complete the defintion of the input
matrix in~(\ref{eq:toherm2}) can be computed by transforming 
to Hermite form of an input matrix of the type shown in~(\ref{eq:toherm1}). 
To compute the Hermite form in~(\ref{eq:toherm2}) we can work modulo
the largest invariant factor of $S_2$.   This is described in
Section~\ref{sec:hermmodular} along with the cost of each transformation.
Note that $C$ and $K$ will have good bounds on their total size because
they are 
blocks of the Hermite form an input matrix that has 
determinant equal in absolute value to $\det S_2$.

In the worst case, the largest invariant factor $s$ of a Smith form
$S$ can have $\log s \in \Omega(\log \det S)$, and existing algorithms
to compute the Hermite forms in~(\ref{eq:toherm1}) and~(\ref{eq:toherm2})
are still too costly for our complexity goal.  In Section~\ref{sec:hermbas}
we develop a new {\em slicing} algorithm that more fully exploits
the presence and structure of $S$ and achieves a cost that depends
on the \emph{average} bitlength $O((\log \det S)/m)$, where $m$ is
the dimension of $S$.

In addition to transformations~3 and~5, the main other computational
step needed for finding the Smith massagers for each of the two
subproblems of our algorithm is to transform a modulus to Smith
form.  
In the case of transformation~2 we use an existing algorithm to
compute a Smith massager $(S_1,V_1)$ for $M_1$ and then set $M_3 =
\colmod(M_2V_1,S_1)$ and $G_1 = \colmod(GV_1,S_1)$, where
$\colmod(\cdot,S)$ denotes matrix multiplication modulo a diagonal
matrix $S$. Transformations~4 and~6 work the same way, each requiring
one Smith massager computation and  one $\colmod$ matrix multiplication.
Beyond computing Smith massagers, all of the computational effort
to achieve the transformations discussed so far is integer matrix
multiplication, albeit of a special type.  In Section~\ref{sec:partial}
we develop procedures which use partial linearization  to do these
$\colmod$ matrix multiplications in the required bit complexity.

The total cost of our Hermite basis algorithm depends on the dimension
parameters $n$ (which is fixed) and $m$, and the precision parameter
$\det H$.  The algorithm constructs the subproblems using a
factorization (\ref{eqn:subprobs}) with $m_1=m_2=m/2$, and $(\det
H_1)(\det H_2) = \det H$.  In Section~\ref{sec:runtime} we give a
complexity analysis that shows that this even partitioning of $m$
but output dependent splitting of $\det H$ corresponds to a recursion
tree that is root-heavy.

The details needed for the case where $M$ is rectangular but of
full column rank requires a different  precision measure and is
described in Section~\ref{sec:general}. In Section~\ref{sec:examples}
we give a number of examples solved by our work. This includes
solving a version of the multivariable Chinese remainder problem and
efficient computation of the Hermite form for a number of interesting
cases.  The paper  ends with a conclusion and topics for future
research.

\subsection{Cost model}  \mylabel{ssec:cost}

The number of bits in the binary representation of an integer $a$ is
$$
\length(a) := \left \{ \begin{array}{ll} 1 & \hbox{if $a=0$} \\
                                  1 +  \lfloor \log_2 |a| \rfloor
    & \hbox{otherwise}
\end{array} \right . .
$$
We will use the fact that $|a| < 2^{\length(a)}$.  Now let $A \in \Z^{n
\times m}$. We define $\length(A) := \length(||A||)$, the bitlength
of the largest entry in absolute value.  Cost estimates will often
be expressed in terms of a bound $D$ for the sum of the bitlengths
of the columns of $A$.  Note that if $A$ is the zero matrix then
we can take $D=m$. If $A$ has entries in each column of magnitude
less than the corresponding diagonal entries of a nonsingular
diagonal matrix $S \in \Z_{\geq 0}^{m \times m}$ then we can take
$D = m + \log_2 \det S$.

Our cost estimates are given in terms of a fixed $2 < \omega \leq
3$ such that two $n \times n$ matrices can be multiplied together
in $O(n^{\omega})$ operations from a commutative ring using only
the operations $\{+, - \times \}$.
The standard algorithm for matrix multiplication  has $\omega=3$,
while the current best known asymptotic upper bound for $\omega$
allows $\omega < 2.37286$~\cite{AlmanWilliams2021}.

Following~\cite[Section~8.3]{vonzurGathenGerhard}, we give cost
estimates using a nondecreasing function  $\M(d) : \Real_{\geq 0}
\rightarrow \Real_{\geq 0}$ together with an algorithm that can
multiply two integers with magnitude bounded by $a\geq 1$ using at
most $\M(\log a)$ bit operations.  We assume $\M(d) = 1$ for  all
$0\leq d \leq 1$.  We make the following assumptions in order to
simplify cost estimates in the analysis of our algorithms.
\begin{itemize}
\item  $\M(d_1) + \M(d_2) \leq \M(d_1 + d_2)$ for all $d_1,d_2 \geq 1$ (superlinearity).
\item $\M(d_1d_2) \leq \M(d_1)\M(d_2)$ for all $d_1,d_2 \geq 0$ (submultiplicativity).
\item $\M(d) \in O(d^{\omega -1 -\epsilon_0})$ for some $0 < \epsilon_0 < \omega-2$. 
\end{itemize}
The last assumption states that if fast matrix multiplication techniques are used, 
then fast integer multiplication should be used as well.
We use $\B(d)$ to bound the cost of integer gcd-related computations
such as the extended euclidean algorithm.  We can always take $\B(d)
= O(\M(d) \log d)$.  The assumptions stated above for $\M$ also apply to $\B$.

\section{Preliminaries}\mylabel{sec:prelim}

In this section we define notation and basic concepts, and recall
from the literature some computational subroutines that we need.

\subsection{Basic notation and concepts}

We identify the residue class ring $\Z/(s)$ with the set
$\{0,1,\ldots,s-1\}$.  Let $S \in \Z_{\geq 0}^{m \times m}$ be diagonal
and $A \in \Z^{\ast \times m}$.  By $\colmod(A,S)$ we mean the
matrix obtained from $A$ by replacing the entries in each column
with their residues modulo the corresponding diagonal entries in $S$.
If $A=\colmod(A,S)$ then we say that \emph{$A$ is reduced
column-modulo $S$}.  The notations $\rowmod(\ast,S)$ and \emph{reduced
row-modulo $S$} are analogous, but for rows.

Let $A$ be a full column rank integer matrix.  The set ${\mathcal
L}(A)$  of all $\Z$-linear combinations of rows of $A$ forms a
lattice.  By \emph{basis of ${\mathcal L}(A)$} we mean a nonsingular
integer matrix $P$ with ${\mathcal L}(P)= {\mathcal L}(A)$.  We
will often simply say \emph{basis of $A$} to mean basis of ${\mathcal
L}(A)$.  By \emph{Hermite basis of $A$} we mean the basis of $A$
that is in Hermite form.  Note that if $A$ is square then the Hermite
basis of $A$ and the Hermite form of $A$ are the same.  Otherwise,
the Hermite basis of $A$ is simply the submatrix of the Hermite
form of $A$ consisting of the nonzero rows.

Let $H \in \Z^{n \times n}$ be a Hermite basis.  We say that \emph{$H$
is index $(k,m)$} if it has the shape
\begin{equation} \mylabel{eq:shape}
H = \left [ \begin{array}{ccc} I_k & \ast  &  \\
 & \bar{H} &  \\
 & &  I_{n-k-m} \end{array} \right ] \in \Z^{n \times n},
\end{equation}
where $\bar{H}$ has $m$ columns, and where integers $k$ and $m$
satisfy $0 \leq k \leq k+m \leq n$.  Note that every Hermite basis
of dimension $n$ is index $(0,n)$ since no additional structure is
imposed.

Let $M$ and $F$ be integer matrices with the same column dimension.  If
\begin{equation} \mylabel{eq:inMF} 
\left[ \begin{array}{c} M \\ F \end{array}
\right] \end{equation}
has full column rank, then we say that \emph{$M$ and $F$ are coprime}
if the only common integer matrix right factor is unimodular.  A
necessary and sufficient condition for $M$ and $F$ to be coprime
is that the Hermite basis of ${\mathcal L}(M) + {\mathcal L}(F)$,
which is the Hermite basis of~(\ref{eq:inMF}), is equal to $I$.

\subsection{Smith massagers} \mylabel{ssec:sm}

As in~\cite{BirmpilisLabahnStorjohann23b} we view the Hermite form
$H$ of a nonsingular matrix $M$ as the unique minimal denominator
of its rational inverse $M^{-1}$.  The key tool used to work
efficiently in this setting is the notion of a Smith massager,
introduced in~\cite{BirmpilisLabahnStorjohann20,
BirmpilisLabahnStorjohann21} for computing the Smith form and
multiplier matrices for a nonsingular matrix.

\begin{definition}[\protect{\cite[Definition~1]{BirmpilisLabahnStorjohann21}}]
 \mylabel{def:sm}
Let $M\in\Znn$ be nonsingular with Smith form $S$. 
The pair of matrices $(S, F)$ is a  \emph{Smith massager} for $M$ if 
$F \in \Znn$ and
\begin{itemize}
\item[(i)]
$MF \equiv 0 \bmod S$, and
\item[(ii)] there exists a matrix $W \in\Znn$ such that
$W F \equiv I_n \bmod S.$
\end{itemize}
\end{definition}
Note that if $(S,F)$ is a Smith massager for a matrix $M$, then we
can replace $F$ with $\colmod(F,S)$.  In this case we say that $(S,
F)$ is {\em reduced}.

\begin{remark}
Note that condition (ii) in Definition~\ref{def:sm} is equivalent to 
$S$ and $F$ being coprime.
\end{remark}
\begin{remark} A Smith massager was
originally defined to be the matrix $F$ alone. For convenience,
since the algorithm we use to compute a Smith massager will
also compute $S$, we define the Smith massager to be the  pair
$(S,F)$.  
\end{remark}
The key feature of a Smith massager that we exploit here is the following.
\begin{lemma} \label{lem:key}
Let $M \in \Z^{n \times n}$ be nonsingular with Smith form $S$.
Then any  Smith massager $(S, F)$ for $M$ has the property that
$FS^{-1}$ has minimal denominator $M$.
\end{lemma}
\begin{proof} 
This follows from \cite[Theorem~4]{BirmpilisLabahnStorjohann21}
where it is  shown that the lattices $$
\{ v \in \Z^{1 \times n} \mid vM^{-1} \in \Z^{1 \times n} \} \mbox{
and  }  \{ v \in \Z^{1\times n} \mid vFS^{-1} \in \Z^{1 \times n}
\}
$$
are identical. 
\end{proof}
As currently defined, a Smith massager $(S,F)$ requires $F$ to be
square.  However, the initial columns of $F$ corresponding to trivial
diagonal elements do not play a role.  In particular,  if $(S,F)$
is reduced and $S$ has $n-m$ trivial diagonals, then we can decompose
$S$ and $F$ as $S = \diag(I_{n-m},\hat{S})$ and $F = \left [
\begin{array}{c|c}\boldz & \hat{F}\end{array} \right ]$.  The
lattices
$$
\{ v \in \Z^{1 \times n} \mid vF S^{-1} \in \Z^{1 \times n} \}  \mbox{ and }
\{ v \in \Z^{1 \times n} \mid v \hat{F} \hat{S}^{-1} \in \Z^{1 \times m} \} 
$$
are identical and hence have the same Hermite basis.  

\begin{example}\mylabel{ex:1}
Consider the matrix $$M = \left[ \begin{array}{ccc}  1 & 2 & 3 \\
4 & 5 & 6 \\ 7 & 8 & 1 \end{array} \right]$$ with Smith form $\diag(1,
1, 24)$.  A reduced Smith massager $(S,F)$ for $M$ is given by
$$S = \left [ \begin{array}{c}24 \end{array} \right ] \mbox{~and~}
F = \left [ \begin{array}{c}
19\\ 10\\3 \end{array} \right ].$$  
Indeed, the conditions of Definition~\ref{def:sm} are satisfied since 
$$
\left[ \begin{array}{ccc}  1 & 2 & 3 \\ 4 & 5 & 6 \\ 7 & 8 & 1
\end{array} \right] \left[ \begin{array}{c} 19 \\ 10 \\3 \end{array}
\right] =  \left[ \begin{array}{c} 48 \\ 144 \\216 \end{array} \right]
\equiv \boldz \bmod 24
\mbox{~~and~~}  
{\mathcal L}(\left  [ \begin{array}{c} 24\end{array} \right ])  +
{\mathcal L}\left(  \left [ \begin{array}{c} 19 \\ 10 \\3
\end{array} \right ] \right) = 
{\mathcal L}(\left [ \begin{array}{c} 1\end{array} \right ]).
$$
The Hermite basis of the lattice $\{ v \in \Z^{1 \times 3} \mid 
v F \equiv 0 \bmod 24 \}$ is
$$
H  = \left[ \begin{array}{ccc} 1& 2& 3\\ & 3 & 6 \\ & & 8 \end{array} \right] ,
$$
which is equal to the Hermite form of $M$.
\end{example}
A reduced Smith massager $(S,F)$ for a Hermite basis $H$ does not
require more space to store than $H$ itself in the worst case.  In
particular, if $H$ has dimension $m$ then it can be represented
using $O(m \log \det H)$ bits by storing only the nontrivial
columns~\cite[Lemma~13]{BirmpilisLabahnStorjohann23b}.  Since $\det
S=\det H$ the same bound holds for representing $(S,F)$.

\begin{theorem}[\cite{BirmpilisLabahnStorjohann20,BirmpilisLabahnStorjohann21}]
\mylabel{thm:costSM} There exists a randomized algorithm
\texttt{SmithMassager}$[\epsilon](M)$ that takes as input a nonsingular
$M \in \Z^{m \times m}$, and returns as output a reduced Smith
massager $(S,F)$ for $M$ with cost $O(m^\omega\, \B(D/m + \log
m)(\log m)^2\log(1/\epsilon))$ bit operations, where $D$ is a bound
for the sum of the bitlengths of the columns of $M$.  The algorithm
is probabilistic of the Las Vegas type: it either returns a correct
result, or reports \textsc{fail} with probability at most $\epsilon$.
\end{theorem}
\begin{proof}
The statement of the theorem is slightly modified
from~\cite[Theorem~19]{BirmpilisLabahnStorjohann21}, which was
specialized to the case $\epsilon = 1/2$, and which gave a cost
estimate with the factor $\B(\log ||M|| + \log m)$.  To obtain the
result in the current theorem we can iterate that algorithm
$\log_2(1/\epsilon)$ times to reduce the chance of returning
\textsc{fail} to $\epsilon$, and use partial
linearization~\cite[Theorem~36]{BirmpilisLabahnStorjohann21} to
replace the  maximum column bitlength $\log ||M||$ with the average
column bitlength $D/m$.
\end{proof}

\subsection{Computation of some  Hermite forms using modular arithmetic}\mylabel{ssec:how}

It is well known that one can compute the Hermite form of a nonsingular
matrix $A \in \Z^{n\times n}$ by working modulo $|\det A|$ to control
intermediate expression swell (see, for example,~\cite{DomichKannanTrotter,
Iliopoulos89:1, HafnerMcCurley, StorjohannLabahnISSAC96}).  In some
cases one can work modulo a divisor $s\in \Z_{>0}$ of $\det A$ by
passing over the Howell form~\cite{Howell}, a natural generalization
to $\Z/(s)$ of the Hermite form over $\Z$.  The following result
applies more generally to a full column rank input matrix.

\begin{lemma}[{\cite[Part~2 of Lemma~7]{BiasseFiekerHofmann17}}]
\mylabel{lem:fieker} Let $A \in \Z^{n \times m}$ and $s \in \Z_{>0}$
be such that $sI  \subseteq {\mathcal L}(A)$. Then the canonical lifting
of the Howell form of $A$ over $\Z/(s^2)$ yields the Hermite form
of $A$ over $\Z$.
\end{lemma}

\begin{theorem} \mylabel{thm:hermcomp} 
Let $A \in \Z^{n \times m}$ and $s \in \Z_{>0}$ be such that $sI
\subseteq {\mathcal L}(A)$.  If $\log ||A|| \in O(\log s)$,
then the Hermite basis $H$ of $A$ can be
computed in $O(nm^{\omega-1}\,\B(\log s))$ bit
operations.
\end{theorem}
\begin{proof}
By Lemma~\ref{lem:fieker} 
the Hermite basis $H$ of $A$ coincides with the first $m$ rows of
the Howell form of $A$ over $\Z/(s^2)$, which can be computed 
in the allotted time~\cite[Section~4]{StorjohannMulders98:1}.
\end{proof}

\begin{example} \mylabel{ex:howell}
Let 
$$M_1 := \left[ \begin{array}{c} A \\ S \end{array} \right]
\mbox{~~and~~} M_2 :=  \left[ \begin{array}{ccc} T & & I \\ B &
 I & \\ S & & \end{array} \right]$$ be integer matrices
with 
$S$ a nonsingular Smith form. Let $s$ be the largest invariant factor of $S$.
Then  $s I = s S^{-1} S = S^\ast S$, with $S^\ast$ an integer
matrix and so $sI \in {\mathcal L}(M_1)$. Similarly
$$
s I =  s {M_2}^{-1} M_2 = \left[ \begin{array}{ccr}  & &  S^\ast \\  &
s I &- B S^\ast    \\ s I  & & - T S^\ast \end{array} \right]
\left[ \begin{array}{ccc} T & & I \\ B & I & \\ S & & \end{array}
\right]
$$
and so $s I \in {\mathcal L}(M_2)$. Thus we can find the Hermite
forms of both $M_1$ and $M_2$ by computing their Howell forms over
$\Z/(s^2)$.  Computing the Hermite forms of matrices with these
shapes are key steps in the Hermite basis algorithm presented in
Section~\ref{sec:alg}.
\end{example}

\section{Bases of integer relations: simplifications and properties } \mylabel{sec:relations}

In this section we show how a given integer relations lattice
$\Rel(M,F)$ can be transformed to a new description that is in some
way simpler, for example with $M$ in Hermite form, or with $M$ and
$F$ being coprime.  These transformations allow us to control the
dimension of the matrices arising during computations along 
with the bitlength of their entries.

\subsection{Integer relations lattices}

Given $M \in \Z^{\ell  \times  m}$ having full column rank and $F
\in \Z^{n \times m}$, the set of integer vectors
\begin{equation} \mylabel{eq:defil}
\Rel(M,F) := \left  \{ p \in \Z^{1 \times n} \mid pF = \boldz \bmod M \right \}
\end{equation}
forms a lattice, which we call the {\em integer relations lattice}
for $F$ with {\em modulus} $M$.  Here the notation $A = \boldz \bmod
M$ stands for ``$A= Q M$ for some integer matrix $Q$,''  that is,
the rows of $A$ are in ${\mathcal L}(M)$.  We have adopted the
notation in~(\ref{eq:defil}) from~\cite{NeigerVu17} where it was
introduced in the context of polynomial matrices.


\begin{remark} The definition of $\Rel(M,F)$ extends naturally to the case when
$M$ and $F$ are over the reals.  For example, $P$ being
a minimal denominator of a rational matrix $R$ is equivalent to 
$P$ being a basis for $\Rel(I,R)$.  In this paper, however, we
stipulate that $M$ and $F$ be integral.
\end{remark}

Lemma~\ref{rem:alt} gives alternate description of $\Rel(M,F)$ that
is useful for Hermite basis computation.

\begin{lemma} \mylabel{rem:alt}  
$\Rel(M,F)$ corresponds to the full column rank input matrix
\begin{equation} \mylabel{eq:alt1}
\left [ \begin{array}{c|c} M &  \\
F & I \end{array} \right ] 
\mbox{~~that has Hermite basis~~}
\left [ \begin{array}{c|c} T & \ast \\
 & H 
 \end{array} \right ] 
\end{equation}
where ${\mathcal L}(T) = {\mathcal L}(M) + {\mathcal L}(F)$ and ${\mathcal L}(H) =\Rel(M,F)$.  
\end{lemma}
\begin{proof}
That ${\mathcal L}(T) = {\mathcal L}(M) + {\mathcal L}(F)$ follows
from the fact that the leading columns of the Hermite basis depend
only the corresponding leading columns of the input matrix.

To establish that ${\mathcal L}(H) =\Rel(M,F)$ we will
show that ${\mathcal L}(H) \subseteq \Rel(M,F)$
and $\Rel(M,F) \subseteq {\mathcal L}(H)$.
From~(\ref{eq:alt1}) we have that $\left [ \begin{array}{c|c}
\phantom{\boldz} & H \end{array} \right ] = -Q  \left [\begin{array}{c|c}
M & \phantom{\boldz} \end{array} \right] + H \left [ \begin{array}{c|c}
F & I \end{array} \right ]$ for some integer matrix $Q$, and thus
$HF=QM$.  This shows that ${\mathcal L}(H) \subseteq \Rel(M,F)$.
Now let $\bar{H}$ be the Hermite basis of $\Rel(M,F)$.  
Then $\bar{H}F = \bar{Q}M$ for some integer matrix $\bar{Q}$ and 
it follows that $\left [ \begin{array}{c|c} 
\phantom{\boldz} & \bar{H} \end{array} \right ] = -\bar{Q}  \left
[\begin{array}{c|c} M & \phantom{\boldz} \end{array} \right] +
\bar{H} \left [ \begin{array}{c|c} F & I \end{array} \right ]$.
This shows that ${\mathcal L}(\bar{H})$, which is equal to $\Rel(M,F)$, is a subset 
of ${\mathcal L}(H)$.  
\end{proof}

\subsection{Transformations to simplify the inputs of $\Rel(M,F)$} \mylabel{ssec:trans}

In this subsection we describe a number of transformations of an
input $\Rel(M,F)$.  All are based on the following general lemma,
which follows from the definition of a relations lattice.  We remark
that in part~\ref{fact1}, $R$ is allowed to be over $\Q$, provided
that $MR$ and $FR$ remain integral.

\begin{lemma} \mylabel{lem:denequiv} $\Rel(M,F)$ is equal to 
\begin{enumerate}
\item \mylabel{fact3} $\Rel(\bar{M},F)$ for any $\bar{M}$
with ${\mathcal L}(\bar{M}) = {\mathcal L}(M)$.
\item \mylabel{fact4} $\Rel(M,F+QM)$ for any $Q$ over $\Z$.
\item \mylabel{fact1} $\Rel(MR,FR)$ for any nonsingular matrix $R$.
\item \mylabel{factext} $\Rel \Biggl ( \left [ \begin{array}{c|c} 
M' & \\
 & M \end{array} \right ], \left [ \begin{array}{c|c} \boldz & F \end{array} \right ] \Biggr )$
for any full column rank $M'$.
\item $\Rel \Biggl ( M_1 , 
\left [ \begin{array}{c} F \\ M_2 \end{array} \right ]  \Biggr )
\left [ \begin{array}{c} I \\ \boldz \end{array} \right ]$, if $M
=
\left [ \begin{array}{c} M_1 \\ M_2 \end{array} \right]
$
    with $M_1$ having full column rank. \mylabel{factTruncate}
\end{enumerate}
\end{lemma}

\begin{proof}
As the first four parts follow directly from definitions we only
include a proof for the last identity. For this case note first
that for any $p \in \Rel(M, F)$ there is a $q = \left[ \begin{array}{c|c}
q_1 & q_2 \end{array} \right]$ such that
    \begin{align*}
        pF &= qM = \left[ \begin{array}{c|c} q_1 & q_2 \end{array} \right]
                    \left[ \begin{array}{c} M_1 \\ M_2 \end{array} \right], \mbox{~~that is,~~} 
                   \left[ \begin{array}{c|c} p & -q_2 \end{array} \right] 
                   \left[ \begin{array}{c} F \\ M_2 \end{array} \right] = q_1M_1.
         \end{align*}
         Thus 
$$ \left[ \begin{array}{c|c} p & -q_2 \end{array} \right] 
                \in \Rel \Biggl ( M_1 , \left [ \begin{array}{c} F \\ M_2 \end{array} \right ]  \Biggr ), 
     \mbox{~~that is,~~} 
    p \in \Rel \Biggl ( M_1 , \left [ \begin{array}{c} F \\ M_2 \end{array} \right ]  \Biggr ) 
              \left [ \begin{array}{c} I \\ \boldz \end{array} \right ].
$$
A similar argument shows that  $p \in  \Rel \Biggl ( M_1 , \left [
\begin{array}{c} F \\ M_2 \end{array} \right ]  \Biggr ) \left [
\begin{array}{c} I \\ \boldz \end{array} \right ]$ also  implies $p \in \Rel(M, F)$.
\end{proof}
We remark that Lemma~\ref{rem:alt} gives an alternate  way to prove the above results. For example,
for  Lemma~\ref{lem:denequiv}.\ref{factTruncate} the corresponding
input matrices from~(\ref{eq:alt1}) are
$$
\Rel (M,F)\sim \left [ \begin{array}{cc} M_1 & \\
M_2 & \\
 F & I \end{array} \right ] \mbox{~~~~and~~~~}
\Rel \Biggl ( M_1 , 
\left [ \begin{array}{c} F \\ M_2 \end{array} \right ]  \Biggr )
\sim
\left [ \begin{array}{cc||c} M_1  & &  \\
F & I &\\ 
M_2  & & I  \end{array} \right ].
$$
Lemma~\ref{lem:denequiv}.\ref{factTruncate} then follows by noting
that the Hermite form of the first input matrix is equal to the
leading columns of the Hermite form of second input matrix.

\subsubsection{Transforming the modulus to Hermite and Smith  forms}

Suppose $T$ is the Hermite basis of the modulus $M$ in $\Rel(M,F)$.
Then compared to $M$, which may have more rows than columns, $T$
is square and nonsingular and has the benefit of being a canonical
presentation of ${\mathcal L}(M)$.  Lemma~\ref{lem:compress} shows
that we can replace $M$ by $T$, and that $F$ can be replaced by its
\emph{remainder $\bar{F}$ with respect to $T$}, the unique matrix
that satisfies $\bar{F} = F \bmod T$ and is reduced column-modulo
the diagonal entries of $T$.  The new inputs $T$ and $\bar{F}$  have
controlled size, with  $T \in \Z^{m \times m}$ and $\bar{F} \in
\Z^{n \times m}$ requiring only $O(m\log \det T)$ and $O(n \log
\det T)$ bits to represent, respectively
(see~\cite[Section~3.3]{BirmpilisLabahnStorjohann23b}).

\begin{lemma} \mylabel{lem:compress} 
The following hold:
\begin{enumerate}
\item \mylabel{compp1}
$\Rel(M,F) = \Rel(T,F)$, where $T$ is the Hermite basis of $M$.
\item \mylabel{compp2}
If $T$ is a Hermite basis, then
 $\Rel(T,F)= \Rel(T,\bar{F})$, where
$$
\left [ \begin{array}{c|c} I & F \\\hline &  T
\end{array}\right]
\mbox{~~has Hermite form~~}
\left [ \begin{array}{c|c} I & \bar{F} \\\hline & T
\end{array}\right].
$$
\end{enumerate}
\end{lemma}
\begin{proof}
Part~\ref{compp1} follows directly  from Lemma~\ref{lem:denequiv}.\ref{fact3}.
For part~\ref{compp2}, note that the unique transformation to Hermite form has the shape
$$
\left [ \begin{array}{c|c} I & Q \\\hline
 & I \end{array} \right ] \left [ \begin{array}{c|c} 
 I & F \\\hline
 & T \end{array} \right ] =
\left [ \begin{array}{c|c}  I & \bar{F} \\\hline
& T \end{array} \right ],
$$
for some integer matrix $Q$.
Then $\bar{F} = F + Q T$ and 
Lemma~\ref{lem:denequiv}.\ref{fact4} gives that $\Rel(T,F) = \Rel(T,\bar{F})$.
\end{proof}
It is also possible to replace the modulus $M$ 
with its Smith form. In this case, however, we also need to modify $F$ appropriately. 
\begin{lemma} 
\mylabel{lem:diag} If $M$ is nonsingular, then
$\Rel(M,F) = \Rel(S,FW)$
where $(S,W)$ is a Smith massager for $M$.
\end{lemma}
\begin{proof}
From~\cite[Theorem~4]{BirmpilisLabahnStorjohann21}  we have that
\begin{equation}\mylabel{eq:pF} 
\{v \in \Z^{1 \times m} \mid v = \boldz \bmod M\} = 
\{v \in \Z^{1 \times m} \mid vW = \boldz \bmod S\}.
\end{equation}
Let $n$ be the row dimension of $F$.
The following equation recalls the definition of $\Rel(M,F)$, and
then applies (\ref{eq:pF}) with $v=pF$.
\begin{eqnarray*}
\Rel(M,F) & = & \{  p \in  \Z^{1 \times n} \mid pF = \boldz \bmod M \} \\
 &= & \{  p \in  \Z^{1 \times n} \mid pFW = \boldz \bmod S \} \\
 &= & \Rel(S,FW). \qedhere
\end{eqnarray*}
\end{proof}

For a rectangular modulus we have the following extension of Lemma~\ref{lem:diag}.

\begin{lemma}
\mylabel{thm:diag} If $M_1$ is nonsingular, then
$$\Rel\left ( \left [ 
\begin{array}{c} M_1 \\ M_2 \end{array} \right ],F\right ) = 
\Rel \left ( 
\left [ \begin{array}{c} S \\ \colmod(M_2W,S) \end{array} \right ],\colmod(FW,S)\right ),$$
where $(S,W)$ is a Smith massager for $M_1$.
\end{lemma}

\begin{proof}
We have
\begin{eqnarray}
\Rel \left ( \left [ 
\begin{array}{c} M_1 \\ M_2 \end{array} \right ],F \right ) & =  &
\Rel \left (  M_1 , \left [ \begin{array}{c} F \\ M_2 \end{array} \right ]\right ) 
   \left [ \begin{array}{c} I \\ \boldz \end{array} \right ] \mylabel{tF1} \\ 
 & = & 
\Rel \left (  S ,  \left [ \begin{array}{c} FW \\ M_2 W \end{array} \right ]\right ) 
   \left [ \begin{array}{c} I \\ \boldz \end{array} \right ] \mylabel{tF2} \\
& = & \Rel \left ( \left [ \begin{array}{c} S \\  M_2 W 
\end{array} \right ],FW \right ) \mylabel{tF3}.
\end{eqnarray}
Line~(\ref{tF1}) follows by applying Lemma~\ref{lem:denequiv}.\ref{factTruncate}
in the forward direction, (\ref{tF2}) follows from Lemma~\ref{lem:diag},
and~(\ref{tF3}) follows by applying Lemma~\ref{lem:denequiv}.\ref{factTruncate} 
in the reverse direction.
By Lemma~\ref{lem:denequiv}.\ref{fact3} and \ref{lem:denequiv}.\ref{fact4}, 
we can replace the matrices $M_2W$ and $FW$ in (\ref{tF3}) with
$\colmod(M_2W,S)$ and $\colmod(FW,S)$, respectively.
\end{proof}

Lemma~\ref{lem:sreduce} shows that we can reduce the column dimension of the
arguments of $\Rel(S,F)$  by removing leading columns
corresponding to trivial invariant factors in $S$. Conversely, we could
start with $\Rel(\hat{S},\hat{F})$ and implicitly add some trivial invariant
factors to $\hat{S}$ and corresponding zero columns to $\hat{F}$
to adjust the column dimension upward.
\begin{lemma} \mylabel{lem:sreduce} 
If $S = \diag(I,\hat{S})$ and $F = \left [ \begin{array}{c|c} \ast & \hat{F} \end{array} \right ]$
then $\Rel(S,F)$ is equal to $\Rel(\hat{S},\hat{F})$.
\end{lemma}
\begin{proof} 
Using Lemma~\ref{lem:denequiv}.\ref{fact4} and~\ref{lem:denequiv}.\ref{factext} in succession gives
$\Rel(S,F)= \Rel(S,\left [ \begin{array}{c|c} \boldz & \hat{F} \end{array} \right ]) = \Rel(\hat{S},\hat{F}).$ 
\end{proof}

\subsubsection{Removing common right divisors}

A relations lattice $\Rel(M,F)$ can be simplified by removing any
common right matrix divisors.  We say that \emph{$\Rel(M,F)$ has
coprime inputs} when $T$ in the following lemma is $I$.

\begin{lemma} \mylabel{lem:reduce}
Let $T$ be the Hermite basis of ${\mathcal L}(M) + {\mathcal L}(F)$.
Then $\Rel(M,F)$ is equal to
$\Rel(MT^{-1}, FT^{-1})$, with the latter having coprime inputs.
\end{lemma}
\begin{proof}
Postmultiplying any full column rank integer matrix by the inverse of a basis for its row lattice
yields an integer matrix whose rows generate~$I$.  Thus $\Rel(MT^{-1},FT^{-1})$ has  inputs which are coprime.
In addition, from Lemma~\ref{lem:denequiv}.\ref{fact1} we have that  $\Rel(M,F)=\Rel(MT^{-1},FT^{-1})$.
\end{proof}
An issue with using  Lemma~\ref{lem:reduce} directly is that the
bitlength of entries in $MT^{-1}$ and $FT^{-1}$ may be large.\footnote{\cite[Example~10]{BirmpilisLabahnStorjohann19}
gives an example of an $m \times m$ Hermite form $T$ with $\log_2
||T|| = 1$ and $\log ||T^{-1}||= \Theta(m)$.}
The following theorem gives an alternative method to construct coprime inputs.
\begin{theorem}\mylabel{thm:bas}
Let $T$  be the Hermite basis of
${\mathcal L}(M) + {\mathcal L}(F)$.
Then
\begin{equation}\mylabel{thm:form}
\left [ \begin{array}{ccc} T & & I \\
F & I &  \\
M & &  \end{array}\right ]
\mbox{~~has Hermite basis with shape~~}
\left [\begin{array}{ccc} T & & \ast \\
  & I & C\\
  & &  K \end{array}\right ]
\end{equation}
and $\Rel(M,F)$ is equal to $\Rel(K,C)$, with the latter
having coprime inputs.
\end{theorem}
\begin{proof}  Since $FT^{-1}$ and $M T ^{-1}$ are integral
the following transformation is unimodular:
\begin{equation} \mylabel{eq:hshape}
\left [ \begin{array}{ccc} I & &  \\
- FT^{-1} & I &  \\
- M T^{-1}  & & I  \end{array}\right ] \left [ \begin{array}{ccc} T & & I \\
F & I &  \\
M & &  \end{array}\right ] = \left [\begin{array}{ccc} T & & I \\
  & I & - F T^{-1} \\
  & &  - M T^{-1}  \end{array}\right ]. 
\end{equation}
The Hermite basis of the matrix on the right hand side of~(\ref{eq:hshape}) has
the shape shown in~(\ref{thm:form}).
By Lemma~\ref{lem:compress} we have $\Rel(K,C) = \Rel(-MT^{-1},-FT^{-1})$, which 
is equal to $\Rel(MT^{-1},FT^{-1})$. The result now follows from Lemma~\ref{lem:reduce}.
\end{proof}

\subsection{Additional properties of bases of a lattice of  integer relations}  \mylabel{ssec:prop}

Ensuring that $\Rel(S,F)$ has coprime inputs links the structure
of the Hermite basis $H$ of $\Rel(S,F)$ with the modulus $S$. 
Theorem~\ref{thm:S} allows us to conclude that $\det H = \det S$. 
Corollary~\ref{cor:atmostm} notes that the number of nontrivial invariant
factors of $S$ is bounded by the number of nontrivial columns of~$H$.

\begin{theorem} \mylabel{thm:S} Assume the inputs to $\Rel(S,F)$
are coprime with $S$ a nonsingular Smith form.  Then, up to trivial
invariant factors, the Smith form of any basis for $\Rel(S,F)$ is
equal to $S$.
\end{theorem}

\begin{proof}
Since all bases for $\Rel(S,F)$ are left equivalent, it will
suffice to prove the result for the Hermite basis $H$ of $\Rel(S,F)$.
The result follows from Lemma~\ref{rem:alt} which observes that
\begin{equation} \mylabel{eq:sh}
\left [ \begin{array}{c|c} S & \\\hline F & I \end{array} \right ]
\mbox{~~has Hermite form~~}
\left [ \begin{array}{c|c} I & \ast \\\hline &  H  \end{array} \right ].
\end{equation}
In particular, the matrices in~(\ref{eq:sh}) are equivalent
to $\diag(I,S)$ and $\diag(I,H)$, respectively.
\end{proof}

\begin{corollary} \mylabel{cor:atmostm}
Assume the inputs to $\Rel(S,F)$ are coprime with $S$ a nonsingular
Smith form.  If the Hermite basis $H$ of $\Rel(S,F)$ is index $(\ast,
m)$, then $S$ has at most $m$ nontrivial invariant factors.
\end{corollary}
\begin{proof}
If $H$ is index $(k, m)$ it can be written using a block decomposition
as shown in~(\ref{eq:shape}).  Given its structure, $H$ is right
equivalent to the block diagonal matrix $\diag(I_k, \bar{H},
I_{n-m-k})$, which has Smith form equal to $\diag(I_{n-m}, \bar{S})$,
where $\bar{S}$ is the Smith form of the $m \times m$ matrix
$\bar{H}$.  The result now follows from Theorem~\ref{thm:S}.
\end{proof}

\section{The Hermite basis algorithm} \mylabel{sec:alg}

In this section we present a recursive method to compute the Hermite
basis $H$ for an $\Rel(S,F)$ with coprime inputs and modulus $S$ a
nonsingular Smith form.  The precondition of coprime inputs implies
that $\det S = \det H$, which will be important for the cost analysis.
We remark that if $(S,F)$ is a Smith massager for a nonsingular
matrix $A$, then the algorithm here directly results in a recursive
algorithm to compute the Hermite form of $A$.

The key idea of our approach is to recursively split the problem
into two subproblems according to the desired output basis $H$.
Because a Hermite basis $H$ is upper triangular, we can choose a
factorization $H=H_2H_1$ such that the product $H_2H_1$ avoids any
computation. For example, when $H$ is $n \times n$ and $n=n_1+n_2$
is a given partition of $n$ we have
\begin{equation} \mylabel{eq:partn}
H = \left [ \begin{matrix} \bar{H}_1 & H_{12} \\
 & \bar{H}_2 \end{matrix} \right] =
{\left [ \begin{matrix} I_{n_1} & H_{12} \\
 & \bar{H}_2 \end{matrix} \right ]}
{
\left [ \begin{matrix} \bar{H}_1 & \\
 & I_{n_2} \end{matrix} \right ]} = H_2 H_1 .
 \end{equation}
At the top level of the recursion
the algorithm will choose the partitioning $n=\lfloor n/2 \rfloor
+ \lceil n/2 \rceil$, so that each subproblem recovers about half
the columns of the Hermite form.  The two recursive subproblems are
to compute first $H_1$ and  then $H_2$.
This divide and conquer approach is based on the following theorem,
which holds generally for any factorization $H_2H_1$ of a basis.
\begin{theorem} \mylabel{thm:rec2} 
If $H_2$ and $H_1$ are nonsingular integer matrices with
$H_2H_1$ a basis for $\Rel(M,F)$, then
\begin{enumerate}
\item \mylabel{Tfact2} $H_1$ is a basis for
\begin{equation} \mylabel{eq:my}
\Rel \left ( \left [ \begin{array}{c} M \\
H_1 F \end{array} \right ] , F \right ).
\end{equation}
\item \mylabel{Tfact1} $H_2$ is a basis for $\Rel(M,H_1F)$.
\end{enumerate}
\end{theorem}
\begin{proof}  Our proof is based on the fact that any two bases
(as nonsingular matrices) are left multiples of each other via a
unimodular matrix, or equivalently,  have the same minimal determinant.
We will implicitly use the fact that $H_1$ and $H_2$ are nonsingular.

First consider part~\ref{Tfact2}.  From the definition of an integer
relations lattice we have that ${\mathcal L}(H_1)$ is subset
of~(\ref{eq:my}). Thus, to show that $H_1$ is a basis for (\ref{eq:my}),
it will suffice to show that $\det H_1 \mid \det T_1$, where $T_1$
is any basis for (\ref{eq:my}).  For any such $T_1$ there exist
integer matrices $Q_1$ and $Q_2$ such that

\begin{equation} \mylabel{eq:i5} 
T_1F = \left [ \begin{array}{cc} Q_1 & Q_2 \end{array} \right ] 
\left [ \begin{array}{c} M \\
H_1 F \end{array} \right ].
\end{equation} 
Rewriting (\ref{eq:i5}) gives
\begin{equation} \mylabel{eq:i6}
(T_1 - Q_2 H_1)F = Q_1 M.
\end{equation}
Since $H_2H_1$ is a basis for
$\Rel(M,F)$, from (\ref{eq:i6}) we 
have ${\mathcal L}(T_1 - Q_2H_1) \subseteq {\mathcal
L}(H_2H_1)$, so there exists a matrix $G$ such that $T_1 - Q_2H_1
= G(H_2H_1)$.  Solving for $T_1$ gives $T_1 = (GH_2+Q_2)H_1$, and hence $\det H_1 \mid  \det T_1$.

Now consider part~\ref{Tfact1}.  Because $H_2H_1$ is a basis for
$\Rel(M,F)$, we have ${\mathcal L}(H_2) \subseteq \Rel(M,H_1F)$.
Thus, to show that $H_2$ is a basis for $\Rel(M,H_1F)$, it will
suffice to show that $\det H_2 \mid \det T_2$, where $T_2$ is
a basis of $\Rel(M,H_1F)$.   For such a basis $T_2$ we have
${\mathcal L}(T_2H_1) \subseteq \Rel(M,F)$.
Because $H_2H_1$ is a basis for $\Rel(M,F)$, we have
$\det H_2H_1 \mid  \det T_2 H_1$, from which it follows
that $\det H_2 \mid \det T_2$.
\end{proof}

\subsection{High level overview of the recursion}

\medskip
Our goal is to compute the Hermite basis $H$ of an input $\Rel(S,F)$.
Our divide and conquer algorithm 
factors $H = H_2 H_1$ as in~(\ref{eq:partn}) and
then has two parts. Part~1 constructs
from $\Rel(S,F)$ a subproblem $\Rel(S_1,F_1)$ that has basis $H_1$.
Part~2 then uses the result of part~1 to construct a subproblem
$\Rel(S_2,F_2)$ that has basis $H_2$.
From a cost point of view, the main challenge is to construct the
two subproblems so that they satisfy the preconditions mentioned
at the beginning of this section, that is, the inputs should be
coprime and the modulus should be in Smith form.  

\subsubsection{Recursion part 1: Determining $H_1$} 
\medskip
Recall that, expressed in terms of matrices, the Hermite basis $H$ of
$\Rel(S,F)$ is is the nonsingular Hermite form with minimal determinant that satisfies
\begin{equation} \mylabel{eq:defH} H F = \boldz \bmod S. \end{equation}
Decompose $F$ as
$$
F = \left [ \begin{array}{c} \bar{F}  \\
A \end{array} \right ],
$$
where $A$ is the last $n_2$ rows of $F$.
We can then rewrite (\ref{eq:defH}) as
\begin{equation} \mylabel{eq:h21f}
H_2 H_1 F =  {\left [ \begin{array}{cc} I_{n_1} & H_{12} \\
 & \bar{H}_2 \end{array} \right ]}
{\left [ \begin{array}{cc} \bar{H}_1 & \\
 & I_{n_2} \end{array} \right ]}
{\left [ \begin{array}{c} \bar{F} \\
A \end{array} \right ]} = \boldz \bmod  S.
\end{equation}
Part~1 is based on the following theorem.
\begin{theorem} \mylabel{thm:phase1} 
$H_1$ is the Hermite basis of 
\begin{equation} 
\Rel\left (\left [ \begin{array}{c} S \\ A \end{array} \right ], F \right ).
\mylabel{eq:first}
\end{equation}
Furthermore, if the inputs to  $\Rel(S,F)$ are coprime, then the inputs for (\ref{eq:first}) are also coprime.
\end{theorem}
\begin{proof}
Theorem~\ref{thm:rec2}.\ref{Tfact2} states that $H_1$ is a basis of 
\begin{equation} \mylabel{eq:same}
\Rel \left (\left [ \begin{array}{c} S \\ H_1 F \end{array} \right ],F \right ).
\end{equation}
Since Hermite forms are canonical, $H_1$ is 
the Hermite basis of~(\ref{eq:same}).

In order to show that the relations lattice~(\ref{eq:same}) equals
that of~(\ref{eq:first}), Lemma~\ref{lem:denequiv}.\ref{fact3}
implies that it will suffice to show that 
\begin{equation} \mylabel{eq:show}
{\mathcal L}\left (\left [ \begin{array}{c} S \\ H_1 F 
 \end{array} \right ] \right )=
{\mathcal L}\left (\left [ \begin{array}{c} S \\ 
A \end{array} \right ] \right ).
\end{equation}
Replacing $H_1F$ in the left hand side of~(\ref{eq:show}) using
$$
H_1 F = 
\left [ \begin{array}{cc} \bar{H}_1 & \\
 & I_{n_2} \end{array} \right ]
{\left [ \begin{array}{c} \bar{F} \\
A \end{array} \right ]} = \left [ \begin{array}{c} \bar{H}_1 \bar{F} \\
A \end{array} \right ],
$$
and adding a block of $n_1$ zero rows to the matrix on the right
hand side of~(\ref{eq:show}), we obtain the equation
\begin{equation} \mylabel{eq:show2}
{\mathcal L}\left (\left [ \begin{array}{c} S \\ \bar{H}_1 \bar{F} \\
A  \end{array} \right ] \right )=
{\mathcal L}\left (\left [ \begin{array}{c} S \\ \boldz \\ 
A \end{array} \right ] \right ),
\end{equation}
which is equivalent to~(\ref{eq:show}).
To show~(\ref{eq:show2}), note that~(\ref{eq:h21f}) 
gives the existence of an integer matrix $Q$ such that
\begin{equation} \mylabel{eq:th}
\left [ \begin{array}{ccc} 
I &  & \\
 Q & I_{n_1} & H_{12} \\
  & & I_{n_2} \end{array} \right ]
\left [ \begin{array}{c}  S \\ \bar{H}_1\bar{F} \\
A \end{array} \right ] = 
\left [ \begin{array}{c} S \\ \boldz \\
A \end{array} \right ] .
\end{equation}
Since the transformation in (\ref{eq:th}) is unimodular, it follows that
(\ref{eq:show2}) holds.

The ``furthermore'' claim in the statement of the theorem is obvious.
\end{proof}

\medskip

Starting with the input~(\ref{eq:first}), part~1 now computes
\begin{equation}\mylabel{process1}
\Rel \left  (  \left [ \begin{array}{c} S \\ A \end{array} \right ], F \right ) \rightarrow \Rel (T,F) \rightarrow \Rel(S_1,F_1).
\end{equation}
The first transformation replaces the modulus in the initial relations
lattice in~(\ref{process1}) with its Hermite basis $T$, while the
second transformation replaces $T$ with its Smith form $S_1$.  By
Theorem~\ref{thm:phase1}, the initial relations lattice in~(\ref{process1})
has coprime inputs, thus $\Rel(S_1,F_1)$ also has coprime inputs,
and by Theorem~\ref{thm:S}, we have $\det S_1= \det H_1$.

\medskip

\subsubsection{Recursion part 2: Determining $H_2$} 

\medskip
The second part of our recursion is
based on  the following theorem.  There the matrix $T$ is the Hermite
basis appearing in~(\ref{process1}), which has already been computed in part~1.
\begin{theorem} \mylabel{thm:phase2}
The following hold:
\begin{enumerate}
\item \mylabel{thm:phase2p1}
$H_2$ is the Hermite basis of $\Rel(S,H_1F)$.
\item \mylabel{thm:phase2p2}
The Hermite basis of
$$
\left [ \begin{array}{c} S \\ H_1 F \end{array}\right ]
$$
is equal to $T$.
\end{enumerate}
\end{theorem}
\begin{proof}
Part~\ref{thm:phase2p1} follows from Theorem~\ref{thm:rec2}.\ref{Tfact1}
combined with the fact that Hermite forms are canonical.
Part~\ref{thm:phase2p2} follows as a corollary of the proof of
Theorem~\ref{thm:phase1}, which shows~(\ref{eq:show}).
\end{proof}

The initial step of part~2  is to compute $B = \colmod(H_1F,S)$.
Then Theorem~\ref{thm:phase2}.\ref{thm:phase2p1} states
that $H_2$ is the Hermite basis of $\Rel(S,B)$. However, 
$S$ and $B$ may not be coprime and so we will still need to do
\begin{equation}\mylabel{process2}
\Rel(S,B) \rightarrow \Rel(K,C) \rightarrow \Rel(S_2,F_2).
\end{equation}
The first transformation is to a $\Rel(K,C)$ having coprime inputs, and
is based Theorem~\ref{thm:bas} which states that
$$
\left [ \begin{array}{ccc} T & & I \\ B & I & \\
S & & \end{array} \right ] \mbox{~~has Hermite basis~~}
\left [ \begin{array}{ccc} T & & \ast \\ & I & C \\ & & K \end{array} \right ],
$$
where $T$ is the Hermite basis of Theorem~\ref{thm:phase2}.\ref{thm:phase2p2}
which has already been computed in part~1.
The second transformation replaces $K$ with its Smith form
$S_2$. Since $(K,C)$ are coprime so are $(S_2, F_2)$. 
As was the case with part~1, 
having $\Rel(S_2,F_2)$ with coprime inputs 
means that $\det S_2 = \det H_2$.

\medskip
\subsubsection{Base case} 
\medskip

The base case for the recursion occurs when two conditions hold:
the inputs $(S,F)$ have column dimension one, and 
the Hermite basis of $\Rel(S,F)$ is known to have at most one nontrivial column.
Example~\ref{ex:1} shows that the latter condition 
does not follow from the former.

\begin{lemma} \mylabel{lem:tbasecase} 
Let $\Rel(sI_1 ,F)$ have coprime inputs with $s \in \Z_{>0}$ and
$F=\colmod(F, sI_1)$.  If the Hermite basis $H$ of $\Rel(sI_1,F)$
is known to be index $(k,1)$, then $H$ can be computed in $O(k \,
\M(\log s) + \B(\log s))$ bit operations, or, more simply, by
$O(n\,\B(\log s))$ bit operations. Here $n$ is the row dimension
of $F$.
\end{lemma}
\begin{proof}
It follows from Theorem~\ref{thm:S} that $\det H = \det sI_1 = s$.
Since $H$ is index $(k,1)$, the matrices $H$ and $F$ in the identity
$HF = \boldz \bmod S$ can be decomposed conformally as
\begin{equation*} \mylabel{eq:basesat}
\stackrel{\textstyle H}{\left [ \begin{array}{c|c|c} I_k & C & \\\hline
& s & \\\hline
 & & I \end{array} \right ]}
\stackrel{\textstyle F}{\left [ \begin{array}{c} F_{1\ldots k,1} \\\hline
 F_{k+1,1} \\\hline 
\boldz \end{array} \right ]} =
\boldz \bmod s.
\end{equation*}
Then $C$ is the unique vector  in $[0,s)^{k \times 1}$ such that
$CF_{k+1,1} = - F_{1\ldots k,1} \bmod s$.  Computing $C$ has the
stated cost.
\end{proof}

\begin{example} \mylabel{examp:alg}
Let $M \in \Z^{3 \times 3}$ be the matrix from Example~\ref{ex:1},
with a Smith massager $(S,F)  = ( \left [ \begin{array}{c}
24\end{array} \right ],  \left [\begin{array}{ccc} 19 & 10 &
3\end{array}\right]^T)$. 
Letting $n_1 = 2$ and $n_2 =1$,
Theorem~\ref{thm:phase1} implies that $H_1$ is the Hermite basis of
$\Rel(\left [ \begin{array}{cc} 24 &3 \end{array} \right ]^T, F)$,
which part~1 transforms to $\Rel(\left [ \begin{array}{c} 3\end{array} \right ], F)$ since
$\left [ \begin{array}{c} 3\end{array} \right ]$
is the Hermite basis of $\left [ \begin{array}{cc}24 & 3\end{array}
\right ]^T$.  The Hermite basis of $\Rel(\left [ \begin{array}{c}
3\end{array} \right ], F)$ is
$$H_1 = \left[
\begin{array}{ccc} 1 & 2 &   \\   & 3 &   \\   &  & 1 \end{array}
\right]. $$ 
For the second part of the recursion we have $B = \colmod(H_1 F , S) =  
\left [ \begin{array}{ccc}  15 & 6 & 3 \end{array} \right ]^T$.
Theorem~\ref{thm:phase2} implies that $H_2$ is the Hermite basis of
$\Rel(S,B) = \Rel(\left [ \begin{array}{c} 24 \end{array} \right ], \left [ \begin{array}{ccc} 15 &6 & 3\end{array} \right ]^T)$,
which does not have coprime inputs. Part~2 transforms 
$\Rel(S,B) \rightarrow
\Rel(\left [ \begin{array}{c} 8 \end{array} \right ], \left [ \begin{array}{ccc} 5 &2 & 1\end{array} \right ]^T)$,
with coprime inputs, and with an index $(3,1)$ Hermite basis $$H_2 = \left[ \begin{array}{ccc}
1 &   &  3\\   & 1 &6  \\   &  & 8 \end{array} \right]. $$ Combining $H_1$ and $H_2$
verifies that  $H = H_2 H_1$ as given earlier in Example~\ref{ex:1}.
\end{example}

\subsection{Algorithm \texttt{HermiteBasis}}

\medskip
Algorithm \texttt{HermiteBasis} is shown in Figure~\ref{fig:hermbas}.
For an input $(S,F,k,m)$, the precondition of the algorithm requires
that the Hermite basis $H$ of $\Rel(S,F)$ is index $(k,m)$.  For
any partition $m = m_1+ m_2$ we can factor $H=H_2H_1$ into an index
$(k,m_1)$ Hermite form $H_1$ and an index $(k+m_1,m_2)$ Hermite
form $H_2$ as follows:
\begin{eqnarray} 
 \overbrace{\left [ \begin{array}{cccc} I_k & \ast & \ast  &  \\
    & \bar{H}_1 & \ast & \\
    &  & \bar{H}_2 &  \\
& & & I_{n-k-m}
\end{array} \right ]}^{\textstyle H} = \overbrace{\left [ \begin{array}{cccc} I_k &      & \ast  &  \\
    &  I_{m_1} &  \ast   & \\
    &  & \bar{H}_2 &  \\
& & & I_{n-k-m}
\end{array} \right ]}^{\textstyle H_2}
 \overbrace{\left [ \begin{array}{cccc} I_k & \ast &       &  \\
    & \bar{H}_1 &      & \\
    &  & I_{m_2} &  \\
& & & I_{n-k-m}
\end{array} \right ]}^{\textstyle H_1} \\
 .\mylabel{eq:partm}
\end{eqnarray}
Note that the partition $m=m_1+m_2$ in~(\ref{eq:partm}) corresponds
to the partition $n=n_1+n_2$ in~(\ref{eq:partn}), with $n_1=k+m_1$
and $n_2 = n-(k+m_1)$.

In addition to the dimension parameter $m$, which the algorithm
reduces by a factor of two for the subproblems, the cost
also depends on the precision parameter $\det S$. 
As explained in the previous subsection, the algorithm ensures that
the inputs to each subproblem are coprime. Although not guaranteeing
an even split of the precision, the fact that the inputs are coprime
does achieve that $\det S = (\det S_1)(\det S_2)$.

\begin{figure}
\noindent
\texttt{HermiteBasis}$[\epsilon](S,F,k,m)$\\
\textbf{Input:} 
\begin{itemize}
\item $S \in \Z^{m \times m}$, nonsingular and in Smith form.
\item $F \in \Z^{n \times m}$, reduced column-modulo $S$.
\item Nonnegative $k$ with $0 \leq  k \leq n-m$.
\item $\epsilon$, a nonzero probability
\item[] \# Precondition: Inputs to $\Rel(S,F)$ are coprime.\\
        \# Precondition: $\Rel(S,F)$ has an index $(k,m)$ Hermite basis.
\end{itemize}
\textbf{Output:} 
\begin{itemize}
\item The Hermite basis $H$ of $\Rel(S,F)$, or \textsc{fail}.
\end{itemize}
\# Note: If any call to \texttt{SmithMassager} returns \textsc{fail}, then return \textsc{fail}.\\
\If $m=1$ \Then \\
\ind{1} \Return the Hermite basis of $\Rel(S,F)$ \hfill \textit{\# Lemma~\ref{lem:tbasecase}}\\
\Else \\
\ind{1} $m_1,m_2 := \lfloor m/2 \rfloor,\lceil m/2 \rceil$\\[1em]
\ind{1} \# Part 1: Compute $H_1$\\
\ind{1} Let $A$ be the last $n-k-m_1$ rows of $F$.\\
\ind{1} $T := $ the Hermite basis of ${\mathcal L}(A) + {\mathcal L}(S)$ 
\hfill \textit{\# Theorem~\ref{thm:hermbas}.\ref{thm:hb1}}\\
\ind{1} $V_1,S_1 := $ \texttt{SmithMassager}$[\frac{\epsilon}{4}](T)$ \hfill \textit{\# Theorem~\ref{thm:costSM}}\\
\ind{1} $F_1 := \colmod(FV_1,S_1)$ \hfill \textit{\# Theorem~\ref{thm:partial}}\\
\ind{1} Decompose $S_1= \diag(I_{m-m_1},\bar{S}_1)$ and 
$F_1 = \left [ \begin{array}{c|c} \boldz & \bar{F}_1 \end{array} \right ]$ conformally.\\
\ind{1} $H_1 := $
 \texttt{HermiteBasis}$[\frac{\epsilon}{4}](\bar{S}_1,\bar{F}_1,k,m_1)$\\[1em]
\ind{1} \# Part 2: Compute $H_2$\\
\ind{1} $B := \colmod(H_1F,S)$ \hfill \textit{\# Corollary~\ref{cor:partial}}\\
\ind{1} $\left [\begin{array}{ccc} T & & \ast \\
  & I & C\\
  & &  K \end{array}\right ] : = $ the  Hermite basis of $
\left [ \begin{array}{ccc} T & & I \\
B & I &  \\
S & &  \end{array}\right ]$   \hfill \textit{\# Theorem~\ref{thm:hermbas}.\ref{thm:hb2}}\\
\ind{1} $V_2,S_2 := $ \texttt{SmithMassager}$[\frac{\epsilon}{4}](K)$ \hfill \textit{\# Theorem~\ref{thm:costSM}}\\
\ind{1} $F_2 := \colmod(CV_2,S_2)$ \hfill \textit{\# Theorem~\ref{thm:partial}}\\
\ind{1} Decompose $S_2= \diag(I_{m-m_2},\bar{S}_2)$ and 
$F_2 = \left [ \begin{array}{c|c} \boldz & \bar{F}_2 \end{array} \right ]$ conformally.\\
\ind{1} $H_2 := $ 
\texttt{HermiteBasis}$[\frac{\epsilon}{4}](\bar{S}_2,\bar{F}_2,k+m_1,m_2)$\\[1em]
\ind{1} \Return $H_2 H_1$
\caption{Algorithm \texttt{HermiteBasis}}
\mylabel{fig:hermbas}
\end{figure}

\subsection{Correctness of Algorithm~\texttt{HermiteBasis}}

\medskip
We split the proof into two parts.  Our first theorem proves
correctness as per the definition of an algorithm,  that is, that
for any input the return value of the algorithm matches the output
specification.  The second theorem bounds by $\epsilon$ the probability
of \textsc{fail} being returned.  For  these results, we implicitly
assume that the algorithm is called with a sequence $[\epsilon](S,F,k,m)$
that satisfies all the specified input conditions.

\begin{theorem} \mylabel{thm:hbcorr}
Algorithm \texttt{HermiteBasis} is correct. That is,  it either returns the Hermite basis of $\Rel(S,F)$, or reports \textsc{fail}.
\end{theorem}

\begin{proof}
The return value in the nonrecursive ``\textbf{if}'' case of the algorithm is correct by definition. 
If any of the calls to
\texttt{SmithMassager} returns \textsc{fail}, then the algorithm correctly 
returns \textsc{fail}.
Assume henceforth that we are in  the recursive
``\textbf{else}'' case, and that none of the calls to \texttt{SmithMassager} 
return \textsc{fail}.  As shown in~(\ref{eq:partm}), we decompose
the Hermite basis $H$ of $\Rel(S,F)$ uniquely as $H=H_2H_1$, where
$H_1$ is index $(k,m_1)$ and $H_2$ is index $(k+m_1,m_2)$.  In the
next two paragraphs, we will show that the inputs
$(\bar{S}_1,\bar{F}_1,k,m_1)$ and $(\bar{S}_2,\bar{F}_2,k+m_1,m_2)$
that are constructed in the ``\textbf{else}'' case 
satisfy the preconditions
of the algorithm, and, moreover, that $H_1$ is the Hermite basis
of $\Rel(\bar{S}_1,\bar{F}_1)$ and $H_2$ is the Hermite basis of
$\Rel(\bar{S}_2,\bar{F}_2)$. 
Correctness of the algorithm will then follow from strong induction
on $m$, with base case $m=1$.

Consider the construction of the first subproblem
$(\bar{S}_1,\bar{F}_1,k,m_1)$.  We claim that
\begin{eqnarray}
{\mathcal L}(H_1) & = &
\Rel \left ( \left [ \begin{array}{c} S \\ A \end{array} \right ], F \right ) \mylabel{eq:y0}\\
 &= & 
\Rel (T,F)   \mylabel{eq:y1} \\
 & = & \Rel(S_1,F_1) \mylabel{eq:y2} \\
 & = & \Rel(\bar{S}_1, \bar{F}_1) \mylabel{eq:y3}.
\end{eqnarray}
Line~(\ref{eq:y0}) follows from Theorem~\ref{thm:phase1}, while
line~(\ref{eq:y1}) follows from Lemma~\ref{lem:compress} (transformation
to Hermite modulus).  Line~(\ref{eq:y2}) follows from Lemma~\ref{lem:diag}
(transformation to Smith modulus).  By Theorem~\ref{thm:phase1},
the input on the right hand side of (\ref{eq:y0}) is coprime.  This
allows us to apply Theorem~\ref{thm:S}, which states that,
up to trivial invariant factors, $S_1$ is the Smith form of $H_1$.
By Corollary~\ref{cor:atmostm}, the fact that $H_1$ 
is index $(k,m_1)$ means that it has at most $m_1$ invariant factors.
This shows that that $S_1$ and $F_1$ can be decomposed as 
$\diag(I_{m-m_1},\bar{S}_1)$ and $F_1 = \left [ \begin{array}{c|c} \boldz & \bar{F}_1 \end{array} \right ]$.
Line~(\ref{eq:y3}) now follows from Lemma~\ref{lem:sreduce}.
By construction, $\bar{S}_1 \in \Z^{m_1 \times m_1}$ is nonsingular
and in Smith form, $\bar{F}_1 \in \Z^{n \times m_1}$ is
reduced column-modulo $\bar{S}_1$, and the inputs to $\Rel(\bar{S}_1,\bar{F}_1)$ are coprime.
Finally, from~(\ref{eq:y3}) we have that
$\Rel(\bar{S}_1,\bar{F}_1)$ has an index $(k,m_1)$ Hermite basis, namely $H_1$.  

\medskip

Now consider the construction of the second subproblem
$(\bar{S}_1,\bar{F}_1,k+m_1,m_2)$. We claim that
\begin{eqnarray}
{\mathcal L}(H_2) & = & \Rel(S,H_1F) \mylabel{eq:x0} \\
 & = & \Rel(S,B) \mylabel{eq:x1} \\
 & = & \Rel(K,C) \mylabel{eq:x2} \\
 & = & \Rel(S_2,F_2)\mylabel{eq:x3} \\
 & = & \Rel(\bar{S}_2, \bar{F}_2) \mylabel{eq:x4}
\end{eqnarray}
Line~(\ref{eq:x0}) follows from Theorem~\ref{thm:phase2}.
Since $B=\colmod(H_1F,S)$, line~(\ref{eq:x1}) 
follows from Lemma~\ref{lem:denequiv}.\ref{fact4}.
Line~(\ref{eq:x2}) follows from Theorem~\ref{thm:bas},
which also states that the inputs to $\Rel(K,C)$ are coprime.
Correctness of lines~(\ref{eq:x3}) and (\ref{eq:x4}), and the remainder
of the proof of correctness of this subproblem construction,
is analogous to that for the first subproblem given above.
\end{proof}

\begin{theorem} \mylabel{thm:hbprob}
Algorithm \texttt{HermiteBasis} returns
\textsc{fail} with probability at most~$\epsilon$.
\end{theorem}
\begin{proof} 
We use strong induction on $m$, with base case $m=1$.
Let $T[\epsilon](m)$ be an upper bound on the 
probability that the algorithm
returns \textsc{fail} with a calling sequence
$[\epsilon](\ast,\ast,\ast,m)$. 
Then $T[\epsilon](1)=0$.   The inductive hypothesis is
$T[\epsilon](\bar{m}) \leq \epsilon$ 
for all $1\leq \bar{m} < m$, some $m > 1$.
We then  obtain
$T[\epsilon](m) \leq 2(\epsilon/4) + T[\epsilon/4](\lfloor m/2 \rfloor)
+T[\epsilon/4](\lceil m/2 \rceil)
  \leq  \epsilon$
by using the inductive hypothesis twice.
\end{proof} 

\section{Modular computation of Hermite forms}\mylabel{sec:hermmodular}

Algorithm \texttt{HermiteBasis} requires computing the Hermite
bases of matrices
\begin{equation} \mylabel{eq:shape1}
M_1 := \left [ \begin{array}{c} A \\
S \end{array} \right ]
\mbox{~~and~~} 
M_2 := \left [ \begin{array}{ccc} T & & I \\
B & I &  \\
S & &  \end{array}\right ],
\end{equation}
where $S \in \Z^{m \times m}$ is a nonsingular Smith form.  Part~1
of of \texttt{HermiteBasis} computes the Hermite basis $T$ of $M_1$.
Part~2 then constructs $B$ to complete the definition of the matrix
$M_2$.  A salient feature of $M_2$ is that $T$ is the Hermite basis
of ${\mathcal L}(B) + {\mathcal L}(S)$.

As pointed out in Example~\ref{ex:howell}, the Hermite bases of
$M_1$ and $M_2$ can be computed using modular arithmetic over
$\Z/(s^2)$ where $s$ is the largest invariant factor in $S$.  In
this section we show that unimodular transformation matrices bringing
$M_1$ and $M_2$ into Hermite form can also be computed over $\Z/(s^2)$.
These transformation matrices are used in the fast 
algorithm develped in Section~\ref{sec:hermbas}.

Our approach is to develop a unified algorithm that can be applied
to both input matrices in~(\ref{eq:shape1}).
To this end, we augment $M_1$ with a square block of zeros 
and then transform it to Hermite form in two steps:
\begin{equation} \mylabel{eq:hbastran}
\left [ \begin{array}{c} \boldz \\
 A \\
S \end{array} \right ] \stackrel{1}{\longrightarrow}
\left [\begin{array}{c} T \\
 A \\
S \end{array} \right ] \stackrel{2}{\longrightarrow}
\left [\begin{array}{c} T \\
\phantom{\boldz} \\
\phantom{\boldz} \end{array} \right ].
\end{equation} 
Once $T$ is computed,
step~2, with $B$ replacing $A$, is used to compute the Hermite form of $M_2$.
Because of the structure of $M_2$, the transformation matrix achieving
step~2 in~(\ref{eq:hbastran}) will reveal the blocks of the Hermite
form of $M_2$ that we need.

We remark that, as usual, $\ast$ is used a placeholder for a scalar/matrix that
exists but does not necessarily need to be computed explicitly. 
As an aid to the reader we do however identify 
the blocks $\ast_1$, $\ast_2$, $\ast_3$ and $\ast_4$ 
as these arise in future constructions.
Our first result gives step~1 in~(\ref{eq:hbastran}).

\begin{theorem} \mylabel{cor:hfcomp}
Let $S \in \Z^{m \times m}$ be a nonsingular Smith form with
largest invariant factor $s$, and let $A \in \Z^{n \times m}$ 
have entries reduced modulo $s$. Then the Hermite basis $T$  of ${\mathcal L}(A) 
+ {\mathcal L}(S)$
together with a matrix $E \in [0,s)^{m \times n}$ such that 
\begin{equation} \mylabel{eq:Ein} 
\left [ \begin{array}{c|c|c} I & E & \ast_1 \\\hline
 & I \\\hline
 & & I \end{array} \right ]  
\left [ \begin{array}{c}  \boldz \\\hline A \\\hline S \end{array} \right ] =
\left [ \begin{array}{c}  T \\\hline A \\\hline S \end{array} \right ] 
\end{equation}
for some $\ast_1 \in \Z^{m \times m}$, can be computed in
$O((n+m)m^{\omega-1}\, \B(\log s))$  bit operations.
\end{theorem}

\begin{proof}  
As pointed out in Example \ref{ex:howell}, $sI \subseteq {\mathcal
L}(S)$ is in the lattice generated by the rows of~(\ref{eq:Ein}).
As such we can use Theorem~\ref{thm:hermcomp} to compute $T$ in the
allotted time. It remains to compute $E$.  The proof of
Theorem~\ref{thm:hermcomp} used the Howell form algorithm
of~\cite{StorjohannMulders98:1} to compute $T$ over $\Z/(s^2)$.
Making use of  \cite[Definition~2 and Corollary~1]{StorjohannMulders98:1},
we observe that if we apply the Howell form algorithm to the input
matrix~(\ref{eq:Ein}) with the blocks $A$ and $S$ switched it will
compute as a side effect two $(n+m) \times(n+m)$ matrices $U$ and
$C$ over $\Z/(s^2)$ that can be written conformally as, and satisfy,
\begin{equation}\mylabel{eq:ucah}
U C = \stackrel{\textstyle}{\left [\begin{array}{c|c} \bar{U} & \\\hline
 & I_{n} \end{array} \right ]}
 \stackrel{\textstyle }{\left [ \begin{array}{c|c} I_m & \bar{C} \\\hline
 & I_n \end{array} \right ]}
\left [ \begin{array}{c} S \\\hline A \end{array} \right ] = 
\left [ \begin{array}{c} T \\\hline A \end{array} \right ] \bmod s^2.
\end{equation}
It follows from~(\ref{eq:ucah}) that $\bar{U}S + \bar{U} \bar{C} A = T 
\bmod s^2$.  Since $sI \subseteq {\mathcal L}(S)$, we can take $E 
:= {\rm mod}(\bar{U}\bar{C},s)$.  Since $\bar{U}$ is $m \times m$ 
and $\bar{C}$ is $m \times n$, $E$ can be computed in the allotted time.
\end{proof}

Next we show how to construct a unimodular transformation to achieve
step~2 in~(\ref{eq:hbastran}):
$$
\left [ \begin{array}{ccc} I & & \\
 C & I & \ast_3 \\
 K & & \ast_4 \end{array} \right ]
\left [ \begin{array}{c} T \\
A \\
S \end{array} \right ]  = 
\left [ \begin{array}{c} T \\
\phantom{\boldz}  \\
\phantom{\boldz} \end{array} \right ] .
$$
We will actually compute a Hermite basis for a slightly more
general input, one whose additional components  will be of use in
a later computation.

\begin{lemma} \mylabel{lem:keyM} Let $S \in \Z^{m \times m}$ be a nonsingular Smith form,
and let $T \in \Z^{m \times m}$ be a Hermite basis such that
$S,A \subseteq {\mathcal L}(T)$. 
Then 
\begin{equation} \mylabel{eq:L1} 
M = \left [\begin{array}{cccc} I & F & & \\
 & T & & I \\
 & A & I & \\
 & S & & 
\end{array} \right ] \mbox{~has Hermite form with shape~}
H = \left [\begin{array}{cccc} I & G & &Q\\
 & T & & \ast \\
 &   & I & C \\
 &   & &  K
\end{array} \right ],
\end{equation}
for new blocks $G$, $Q$, $C$ and $K$. Moreover
\begin{equation}\mylabel{eq:L3}
\left [ \begin{array}{cc} I & C \\
 & K \end{array} \right ] \mbox{~is the Hermite form of~} 
\left [ \begin{array}{cc} I & -AT^{-1} \\
 & -ST^{-1} \end{array} \right ]
\end{equation}
and
\begin{equation} \mylabel{eq:Hshape}
\bar{U} M = 
\left [\begin{array}{cccc} I & Q & & \ast_2\\
 & I & &   \\
 & C & I & \ast_3 \\
 & K & & \ast_4 
\end{array} \right ]
\left [\begin{array}{cccc} I & F & & \\
 & T & & I  \\
 & A & I & \\
 & S & & 
\end{array} \right ] =
 \left [\begin{array}{cccc} I & G & &Q\\
 & T & & I \\
 &   & I & C \\
 &   & &  K
\end{array} \right ] = \bar{H},
\end{equation}
where $\bar{U}$ is unimodular.
\end{lemma}
\begin{proof} We will begin by establishing the existence of a
unimodular transformation $\bar{U}$ that transforms $M$ to a matrix
$\bar{H}$, with both $\bar{U}$ and $\bar{H}$ having the shape shown
in~(\ref{eq:Hshape}).  Along the way we will establish that~(\ref{eq:L3})
holds.  Finally, we will do one more transformation on $\bar{H}$
to put it into Hermite form, showing that~(\ref{eq:L1}) holds.

Since $S,A \subseteq {\mathcal L}(T)$, the matrices $ST^{-1}$ and
$AT^{-1}$ are integral and so we can  triangularize $M$~as
\begin{equation} \mylabel{eq:T1}
\left [\begin{array}{cccc} I &   & & \\
 & I & &   \\
 & -AT^{-1} & I & \\
 & -ST^{-1} & & I
\end{array} \right ]
\stackrel{\textstyle}{
\left [\begin{array}{cccc} I & F & & \\
 & T & & I \\
 & A & I & \\
 & S & & 
\end{array} \right ]}  =
\left [\begin{array}{cccc} I & F & & \\
 & T & & I \\
 &   & I & -AT^{-1}\\
 &   & & -ST^{-1} 
\end{array} \right ]
\end{equation}
with  the right hand side of~(\ref{eq:T1}) left equivalent to $M$.  The transformations to Hermite
form of the leading and trailing submatrices of the matrix on the right hand side of~(\ref{eq:T1}) 
have the shape
\begin{equation} \mylabel{eq:T3}
\left [ \begin{array}{cc} I & \ast \\
 & I \end{array} \right ] 
\left [ \begin{array}{cc} I & F \\
 & T \end{array} \right ]  
= 
\left [ \begin{array}{cc} I & G \\
 & T \end{array} \right ]
\end{equation}
and
\begin{equation} \mylabel{eq:T2}
\left [ \begin{array}{cc} I & \ast \\
 & \ast \end{array} \right ] 
\left [ \begin{array}{cc} I & -AT^{-1} \\
 & -ST^{-1} \end{array} \right ]  
= 
\left [ \begin{array}{cc} I & C \\
 & K \end{array} \right ].
\end{equation}
Applying the transformations~(\ref{eq:T3}) and~(\ref{eq:T2}) to
equation~(\ref{eq:T1}) we obtain
\begin{equation} \mylabel{eq:T4}
\left [\begin{array}{cccc} I & \ast  & & \\
 & I & &   \\
 &C & I & \ast\\
 & K & & \ast
\end{array} \right ]
\left [\begin{array}{cccc} I & F & & \\
 & T & & I \\
 & A & I & \\
 & S & & 
\end{array} \right ]  =
\left [\begin{array}{cccc} I & G & & \ast \\
 & T & & I \\
 &   & I & C\\
 &   & &  K
\end{array} \right ].
\end{equation}
The final step to transform the right hand side of~(\ref{eq:T4})
to $\bar{H}$ is to add a  multiple of the last block-row (containing
$K$) to the first block-row in order to reduce entries in the top
right block $\ast$.  Performing this transformation on~(\ref{eq:T4})
gives equation~(\ref{eq:Hshape}).  Since all the transformations
performed were unimodular, so is $\bar{U}$.

Finally, notice that $\bar{H}$ in~(\ref{eq:Hshape}) may not be completely
in Hermite form. In particular, if some diagonal
entries of $K$ are equal to one, then we need to add a multiple of the  last
block-row (containing $K$)  to the second  block-row in order to
ensure all entries above the diagonals in the last block-column
are reduced.  Once we do this we have transformed $\bar{H}$ to its Hermite form, 
which has the shape of $H$ shown in~(\ref{eq:L1}).
\end{proof}

Regarding the input matrix $M$ of~(\ref{eq:L1}), the blocks $T$ and $S$ are
always square of dimension $m$, but the number of rows in $F$ and
$A$ can be arbitrary.  In particular, $F$ or $A$ can be tall or
short, and even be the $0 \times m$ matrix.  
In the three cases
where $F$ is empty,  $A$ is empty, and both are empty, the shape
of $M$ simplifies to
$$
 \left [ \begin{array}{ccc} T & & I \\ A & I \\ 
  S & \end{array} \right ], \mbox{~~}
 \left [ \begin{array}{ccc} I & F & \\
 & T & I \\
 & S & \end{array} \right ] \mbox{~~and~~}
\left [\begin{array}{cc} T &  \\
S & I \end{array} \right ].$$
For convenience, we will use a common bound $n$ for the row dimension of both $F$ and $A$.
The dimension of $M$ is then bounded by $2n+2m$.

\begin{theorem}  \mylabel{thm:compblocks} Let $S$, $T$ and $M$ be
as in Lemma~\ref{lem:keyM} and  $s$ be the largest invariant factor
of $S$.  If entries in $A$ and $F$ are reduced modulo $s$, and $n$
is a common bound for the row dimension of $A$ and $F$, then the
the Hermite form $H$ of $M$ can be computed computed in
$O((n+m)m^{\omega-1} \, \B(\log s))$ bit operations.
\end{theorem}
\begin{proof} 
Note first that the Smith form of $M$ is $\diag(I,S)$, and so $sI
\in {\mathcal L}(M)$.  If $n = 0$ then we can directly apply
Theorem~\ref{thm:hermcomp} to compute  the Hermite form of the $2m
\times 2m$ matrix
$$
\left [ \begin{array}{cc} T & I  \\
S &  \end{array} \right ].
$$
Otherwise, if  $A$ is not empty then we can decompose $A$ as
$$
A = \left [ \begin{array}{c} A_1 \\
 A_2 \\
\vdots \\
A_k  \end{array} \right ]
$$
where all the blocks have $m$ rows, except for possibly the last block $A_k$.  
Note that each
$$
\left[ \begin{array}{cccc} T & & I \\
  A_i & I & \\
  S & &   \end{array} \right ] \mbox{~~has Hermite form~~}
\left[ \begin{array}{cccc} T & & \ast \\
      &  & C_i \\
       & &K  \end{array} \right ].
$$
Thus the submatrix $C$ of $H$ can be found using at most
$k \leq \lceil n/m\rceil$ Hermite form computations of matrices with
dimension bounded by $3m$.   This is within the allotted cost.
The submatrices $G$ and $Q$ of the Hermite form can be found 
by using a similar block decomposition of $F$.
\end{proof}

\section{Faster modular computation of Hermite forms} \mylabel{sec:hermbas}

This section is devoted to establishing that the two Hermite basis
computations required by Algorithm~\texttt{HermiteBasis} can be
done in the complexity that we require.
\begin{theorem} \mylabel{thm:hermbas}
Let $S \in \Z^{m \times m}$ be a nonsingular Smith form and  $A \in
\Z^{n \times m}$ be reduced column-modulo $S$. The following
can be computed in $O((n+m)m^{\omega-1}\, \B((\log \det S)/m + \log m))$
bit operations.
\begin{enumerate}
\item \mylabel{thm:hb1} The Hermite basis $T$ of ${\mathcal L}(A) + {\mathcal L}(S)$.
\item \mylabel{thm:hb2}  The blocks $C$ and $K$ of the Hermite form
$$
\left [ \begin{array}{ccc} T & & \ast \\
 & I & C \\
 & & K \end{array} \right ] \mbox{~~of the input matrix~~}
\left [\begin{array}{ccc} T & & I \\ A & I & \\ S & & \end{array} \right ].$$
\end{enumerate}
\end{theorem}
Section~\ref{ssec:overview} gives a high level
overview of the algorithm supporting Theorem~\ref{thm:hermbas}.
The key computational steps are detailed in Sections~\ref{ssec:construct} and~\ref{ssec:apply}.
We defer the proof of Theorem~\ref{thm:hermbas} to the end of
Section~\ref{ssec:apply}.  

\subsection{High level overview of the algorithm} \mylabel{ssec:overview} 

Our approach is to initialize a work matrix to be equal to
\begin{equation} \mylabel{eq:shape2}
\left [ \begin{array}{cc} 
\boldz \\\hline
A  \\\hline
S 
\end{array} \right ] \in \Z^{(2m+n)\times m}
\end{equation}
and then apply a sequence of unimodular row transformations to
transform it to its Hermite form in stages for $k=1,2,\ldots,\lceil
\log_2 (m+1)\rceil$. 
Stage~1 transforms the work matrix so that the first $m/2$ columns
are in Hermite form. Stage~2 extends the number of columns in Hermite form
to the first $m/2+m/4$,
stage~3 to the first $m/2+m/4+m/8$, and so on.
At the beginning of stage $k$, at most
$m/2^{k-1}$ columns remain to be put into Hermite form.

\begin{example} \mylabel{ex:hermbas}
Suppose our input matrices have sizes $n=3$ and $m=7$, that is, $A \in \Z^{3 \times 7}$
and $S \in \Z^{7 \times 7}$. Then the
work matrix at the start, and at end of stages~1 and~2, has the following shape:
\begin{equation}  
\setlength{\arraycolsep}{.3\arraycolsep} \renewcommand{\arraystretch}{.7}
\left [ \begin{array}{||cccc|ccc}
  0    &        &        &        &        &         &        \\
       &   0    &        &        &        &         &        \\
       &        &   0    &        &        &         &        \\
       &        &        &   0    &        &         &        \\
       &        &        &        &   0    &         &        \\
       &        &        &        &        &    0    &        \\
       &        &        &        &        &         &       0\\\hline
 \ast  &   \ast &   \ast &   \ast &   \ast &   \ast  & \ast   \\
 \ast  &   \ast &   \ast &   \ast &   \ast &   \ast  & \ast   \\
 \ast  &   \ast &   \ast &   \ast &   \ast &   \ast  & \ast   \\\hline
  s_1  &        &        &        &        &         &        \\
       &   s_2  &        &        &        &         &        \\
       &        &  s_3   &        &        &         &        \\
       &        &        &   s_4  &        &         &        \\
       &        &        &        &   s_5  &         &        \\
       &        &        &        &        &   s_6   &        \\
       &        &        &        &        &         &   s_7  
\end{array} \right ]  \stackrel{1}{\rightarrow}
\left [ \begin{array}{cccc||cc|c}
\underline{\ast} &   \ast &   \ast &   \ast &   \ast &   \ast  &  \ast  \\
       & \underline{\ast}  &  \ast  &  \ast  &  \ast  &  \ast   &  \ast  \\
       &        & \underline{\ast}   & \ast   &  \ast    & \ast    &  \ast  \\
       &        &        & \underline{\ast}   &  \ast &   \ast  &   \ast \\
       &        &        &        &   0    &         &        \\
       &        &        &        &        &   0     &        \\
       &        &        &        &        &         &    0   \\\hline
       &        &        &        &   \ast &   \ast  & \ast   \\
       &        &        &        &   \ast &   \ast  & \ast   \\
       &        &        &        &   \ast &   \ast  & \ast   \\
\phantom{s_1} &        &        &        &   \ast &  \ast   & \ast   \\
       & \phantom{s_2} &        &        &  \ast  &  \ast   & \ast   \\
       &        &\phantom{s_3} &    &  \ast  &  \ast   & \ast   \\
       &        &        & \phantom{s_4}   &   \ast &  \ast   &  \ast  \\\hline
       &        &        &        &   s_5  &         &        \\
       &        &        &        &        &   s_6   &        \\
       &        &        &        &        &         &   s_7  
\end{array} \right ] \stackrel{2}{\rightarrow}
\left [ \begin{array}{cccccc||c|}
\underline{\ast} &   \ast &   \ast &   \ast &   \ast &   \ast  &  \ast  \\
       & \underline{\ast}  &  \ast  &  \ast  &  \ast  &  \ast   &  \ast  \\
       &        & \underline{\ast}   & \ast   &  \ast    & \ast    &  \ast  \\
       &        &        & \underline{\ast}   &  \ast &   \ast  &   \ast \\
       &        &        &        &  \underline{\ast}  & \ast    &  \ast  \\
       &        &        &        &        &  \underline{\ast}  &  \ast  \\
       &        &        &        &        &         &    0   \\\hline
       &        &        &        &        &         & \ast   \\
       &        &        &        &        &         & \ast   \\
       &        &        &        &        &         & \ast   \\
\phantom{s_1} &        &        &        &        &         & \ast   \\
     & \phantom{s_2}&        &        &        &         & \ast   \\
     &        & \phantom{s_3} &        &        &         & \ast   \\
     &        &        &\phantom{s_4} &        &         &  \ast  \\
     &        &        &        & \phantom{s_5}&         & \ast   \\
     &        &        &        &        & \phantom{s_6} & \ast   \\\hline
       &        &        &        &        &         &   s_7  
\end{array} \right ].
\mylabel{eq:exhermbas}
\end{equation}
Note that the final stage $\stackrel{3}{\rightarrow}$, which
transforms the last column containing $s_7$ to correct form, is not
shown in~(\ref{eq:exhermbas}).
Note that each stage in~(\ref{eq:exhermbas}) uses both steps in~(\ref{eq:hbastran}), that is,
computing a square Hermite form, eliminating the entries below it, and reducing the entries above it.
\end{example}

We now specify the transformations of the work matrix at the start of each stage in more detail. 
Let $\bar{m}$ be the number of columns remaining to be put into correct form.
Partition $S=\diag(S_1,\bar{S})$, where $\bar{S}$ has column dimension
$\bar{m}$. 
(For example, in~(\ref{eq:exhermbas}) the second matrix has $\bar{m}=3$ with
$\bar{S} = \diag(s_5,s_6,s_7)$.)
Then the work matrix can be written as
\begin{equation} \mylabel{eq:decomp1}
\setlength{\arraycolsep}{1\arraycolsep} \renewcommand{\arraystretch}{1.2}
\left [ \begin{array}{c||c} T_1 & \bar{F} \\
 & \boldz \\\hline
  & \bar{A}\\\hline
  &  \bar{S} \end{array} \right ],
\end{equation}
where $T_1$ is a Hermite basis with $m - \bar{m}$ rows, $\bar{A}$
has row dimension $n+m-\bar{m}$ and  the matrices $\bar{F}$ and
$\bar{A}$ are reduced column-modulo $\bar{S}$.  The row dimension
of $\bar{F}$ and $\bar{A}$ are thus bounded by $m$ and $n+m$,
respectively, independent of the particular stage.  At the start
of stage~1, $\bar{m}=m$ and~(\ref{eq:decomp1}) coincides with the
input matrix~(\ref{eq:shape2}). At the end of the last stage,
$\bar{m}=0$ and~(\ref{eq:decomp1}) is in Hermite form.

If $\bar{m}>0$, then our
goal at the current stage is 
to transform the first half of 
the last $\bar{m}$ columns of (\ref{eq:decomp1}) to
Hermite form.  To this end,  decompose 
$\bar{S} = \diag(\bar{S}_1,\bar{S}_2)$, 
where $\bar{S}_1$ has dimension 
$\bar{m}_1 := \lceil \bar{m}/2 \rceil$, and $\bar{S}_2$ has dimension 
$\bar{m}_2 := \lfloor \bar{m}/2 \rfloor$.  
Conformally decompose $\bar{F} = \left [ \begin{array}{c|c}
\bar{F}_1 & \bar{F}_2 \end{array} \right ]$ and $\bar{A} = \left
[\begin{array}{c|c} \bar{A}_1 & \bar{A}_2 \end{array} \right ]$.
Then we can refine the decomposition of the work matrix in~(\ref{eq:decomp1}) as
\begin{equation} \mylabel{eq:shapek1}
\setlength{\arraycolsep}{1\arraycolsep} \renewcommand{\arraystretch}{1.2}
\left [ \begin{array}{c||c|c}
 T_1 & \bar{F}_1 & \bar{F}_2 \\
 & \boldz &        \\
 & & \boldz \\\hline
 & \bar{A}_1 & \bar{A}_2 \\\hline
&  \bar{S}_1 & \\\
 & & \bar{S}_2 \end{array} \right ] .
\end{equation}
The current stage now applies a pair of unimodular transformations
to (\ref{eq:shapek1}):
\begin{equation} \mylabel{eq:appU}
\setlength{\arraycolsep}{.5\arraycolsep} \renewcommand{\arraystretch}{1.2}
\left [ \begin{array}{ccc|c|cc}
 ~I~   & Q    &   &     &   \ast_2  &   \\
     &  ~I~  &     &     &          &   \\
     &    &~I~   &     &    &   \\\hline
     &    C     & &  ~I~  &  \ast_3  &   \\\hline
     &     K     &&     &  \ast_4 &   \\
     &      &     &     &    & ~I~
\end{array} \right ]
\left [ \begin{array}{ccc|c|cc}
 ~I~   &      &     &     &    &   \\
     &  ~I~  &     & E &  \ast_1  &   \\
     &      & ~I~  &     &    &   \\\hline
     &      &     &  ~I~  &    &   \\\hline
     &      &     &     &  ~I~ &   \\
     &      &     &     &    & ~I~ 
\end{array} \right ]
\left [ \begin{array}{c||c|c}
 T_1 & \bar{F}_1 & \bar{F}_2 \\
 & \boldz &        \\
 & & \boldz \\\hline
 & \bar{A}_1 & \bar{A}_2 \\\hline
&  \bar{S}_1 & \\\
 & & \bar{S}_2 \end{array} \right ] =
\left [ \begin{array}{c||c|c}
T_1 & \bar{G}_1 & \ast \\
 & \bar{T}_1 & \ast   \\
 & &  \boldz \\\hline
 &    & \ast  \\\hline
&            & \ast \\\
 & & \bar{S}_2 \end{array} \right ].
\end{equation}
The right hand side of~(\ref{eq:appU}) has the shape required for
the next stage.  We remark that the blocks $\ast_1$, $\ast_2$,
$\ast_3$ and $\ast_4$ encode unimodular row operations used to keep
entries reduced modulo the diagonal entries of $\bar{S}_1$.  We
will need to show the existence of these blocks, but do not need to
compute them explicitly since the blocks above $\bar{S}_2$ in the
transformed matrix do not depend on them.  

Before detailing the construction of the transformation, and the computation of the right
hand side of (\ref{eq:appU}), we explain how this incremental
approach of putting only half of the remaining columns into correct
form allows us to obtain an overall cost estimate that depends on
the average bitlength $(\log \det S)/m$ of the invariant factors
of $S$.  We exploit a
\begin{equation} \mylabel{eq:inv} 
\hbox{~dimension~} \times \hbox{~precision~} \leq 2(\log \det S)
\end{equation}
invariant that holds for each stage.
The ``dimension'' in (\ref{eq:inv}) refers to the column dimension
$\bar{m}$ of $\bar{S}=\diag(\bar{S}_1,\bar{S}_2)$, 
and ``precision'' refers to the bitlength of the largest invariant
factor $s$ of $\bar{S}_1$.  In particular, the needed blocks of the
unimodular transformation in (\ref{eq:appU}) will be computed modulo~$s$.
\begin{lemma} \mylabel{lem:smallsnew} 
Let $S \in \Z^{m \times m}$ be a nonsingular Smith form.
Partition 
\begin{eqnarray*}
S & = & \diag(S_1,\bar{S})\\
 & = & \diag(S_1,\diag(\bar{S}_1,\bar{S}_2))
\end{eqnarray*}
with the dimension $\bar{m}$ of $\bar{S}$ arbitrary, but the second
partition splitting $\bar{S}$ evenly, with $\bar{S}_1$ of dimension
$\lceil \bar{m}/2 \rceil$.  Let $s$ be the largest invariant factor
of $\bar{S}_1$.  Then $\log s \leq  2(\log \det S)/\bar{m}$.
\end{lemma}
\begin{proof} $\bar{S}_2$ has dimension $\bar{m}-\lceil \bar{m}/2
\rceil = \lfloor \bar{m}/2 \rfloor$.  Since $s$ is in $\bar{S}_1$,
and $s$ divides all invariant factors in $\bar{S}_2$, there are at
least $\lfloor \bar{m}/2 \rfloor +1 \geq \bar{m}/2$ invariant factors
in $\bar{S}$ that are multiples of $s$.  Since $\det S$ is the
product of all invariant factors, we have $s^{\bar{m}/2} \leq \det
S$.  Solving for $\log s$ gives the result.
\end{proof}

\subsection{Constructing the unimodular transformations} \mylabel{ssec:construct} 

We now show how to construct the unimodular transformations
in~(\ref{eq:appU}).  The transformations depend only on the columns
containing $\bar{S}_1$.  Also, note that the $\bar{m}_2$ rows
containing $\bar{S}_2$, and the $\bar{m}_2$ rows containing the
block $\boldz$ directly above $\bar{A}_2$, are not modified, so we
can ignore these rows for now.  The transformation we need to compute
is $U_2U_1$ satisfying
\begin{equation} \mylabel{eq:lemcompa}
\renewcommand{\arraystretch}{1.2}
\stackrel{\textstyle U_2}{
\left [ \begin{array}{cc|c|c} 
I&Q  & & \ast_2 \\ 
&I &    &  \\\hline
& C   & I & \ast_3  \\\hline
 & K   &   & \ast_4 \end{array} \right ]}
\stackrel{\textstyle U_1}{
\left [ \begin{array}{cc|c|c} I &  &   \\
 & I & E & \ast_1 \\\hline
  & & I &  \\\hline
    &  &   & I      \end{array} \right ]} 
\left [ \begin{array}{c} \bar{F}_1 \\
 \boldz \\\hline
\bar{A}_1\\\hline
\bar{S}_1
\end{array} \right ]=
\left [ \begin{array}{c} 
\bar{G}_1 \\
\bar{T}_1       \\\hline
\phantom{\boldz} \\\hline
\phantom{\boldz} \end{array} \right ],
\end{equation}
where the input blocks $\bar{F}_1$ and $\bar{A}_1$ are reduced column-modulo $\bar{S}_1$,
and the output blocks in the right hand side are such that
$$
\left [ \begin{array}{cc} I & \bar{G}_1 \\
 & \bar{T}_1 \end{array} \right]
$$
is in Hermite form.  Recall that $\bar{m}= \bar{m}_1 + \bar{m}_2$.
Comparing with (\ref{eq:decomp1}), we see that $\bar{F}_1$, $\bar{A}_1$
and $\bar{S}_1$ have row dimension equal to $m-\bar{m}$, $n+m-\bar{m}$
and $\bar{m}_1$, respectively.  Instead of these exact dimensions,
for cost estimates we will uses the upper bounds $m$,
$n+m$ and $\bar{m}$, respectively.  More generally, at any stage,
the dimensions of the diagonal blocks of $U_1$ and $U_2$ are bounded
by the dimension of the corresponding diagonal blocks in
$\diag(I_{m}, I_{\bar{m}},I_{n+m},I_{\bar{m}})$.
For example, $E$ has dimension bounded by $\bar{m} \times (n+m)$. By using these upper bounds, the only parameters 
that are specialized to the current stage for the purpose of the cost analysis are $\log s$, where $s$ is the largest
invariant factor of $\bar{S}_1$, and $\bar{m}$.

Observe that the first transformation
\begin{equation} \mylabel{eq:lemcompb}
\renewcommand{\arraystretch}{1.2}
\stackrel{\textstyle U_1}{
\left [ \begin{array}{cc|c|c} I &  &   \\
 & I & E & \ast_1 \\\hline
  & & I &  \\\hline
    &  &   & I      \end{array} \right ]} 
\left [ \begin{array}{c} \bar{F}_1 \\
 \boldz \\\hline
\bar{A}_1\\\hline
\bar{S}_1
\end{array} \right ]=
\left [ \begin{array}{c} \bar{F}_1 \\\hline
 \bar{T}_1 \\ \hline
\bar{A}_1\\ \hline
\bar{S}_1
\end{array} \right ]
\end{equation}
simply computes the Hermite form $\bar{T}_1$ along with the matrix
$E$ which puts the matrix into this form.  
The second transformation $U_2$ then continues to
transform the right hand side
of~(\ref{eq:lemcompb}) to the right hand side of~(\ref{eq:lemcompa}).

\begin{lemma}  \mylabel{lem:getU12}
There exists an algorithm that computes matrices $E$, $Q$, $C$ and $K$ creating unimodular matrices $U_1$ and $U_2$ 
satisfying~(\ref{eq:lemcompa}) with cost
$O((n+m)\bar{m}^{\omega-1}\, \B(\log s))$
bit operations, where $s$ is the largest invariant factor of $\bar{S}_1$.   The algorithm computes
in the same time the blocks $\bar{G}_1$ and $\bar{T}_1$ on the right hand side of~(\ref{eq:lemcompa}).
All the computed blocks have entries from $[0,s]$.
\end{lemma}

\begin{proof}
Theorem~\ref{cor:hfcomp} gives the result for computing $E$ and $\bar{T}_1$.
By Lemma~\ref{lem:keyM}, 
\begin{equation} \mylabel{yy1}
\left  [ \begin{array}{cccc} I & \bar{F}_1 & & \\
 & \bar{T}_1 & & I \\
& \bar{A}_1 & I & \\
 & \bar{S}_1 &  \end{array} \right ]
\mbox{~~has Hermite form with shape~~}
\left [  \begin{array}{cccc}
I & \bar{G}_1 & & Q \\
 & \bar{T}_1  &  & \ast \\
 &  & I & C \\
 & & & K  \end{array} \right ],
\end{equation}
with the blocks $\bar{G}_1$, $Q$, $C$ and $K$ satisfying the requirements in the statement of the lemma.
By Theorem~\ref{thm:compblocks} the Hermite form in~(\ref{yy1}) can be computed in the allotted time.
\end{proof}

\subsection{Applying the unimodular transformations} \mylabel{ssec:apply} 

In order to complete the current stage, we need to find the last $\bar{m}_2$ 
columns of the matrix on the right of (\ref{eq:appU}), which are
given by
\begin{equation} \mylabel{eq:appU2}
\renewcommand{\arraystretch}{1.2}
\stackrel{\textstyle U_2}{
\left [ \begin{array}{cc|c|c} 
I&Q  & & \ast_2 \\ 
&I &    &  \\\hline
& C   & I & \ast_3  \\\hline
 & K     &   & \ast_4 \end{array} \right ]}
\stackrel{\textstyle U_1}{
\left [ \begin{array}{cc|c|c} I &  &   \\
 & I & E & \ast_1 \\\hline
  & & I &  \\\hline
    &  &   & I      \end{array} \right ]} 
\left [ \begin{array}{c} \bar{F}_2 \\
\phantom{\boldz} \\\hline
\bar{A}_2\\\hline
\phantom{\boldz}
\end{array} \right ]=
\left [ \begin{array}{r} 
\bar{F}_2 + QE\bar{A}_2  \\
E\bar{A}_2\\\hline
\bar{A}_2 + CE\bar{A}_2 \\\hline
K  E\bar{A}_2 \end{array} \right ].
\end{equation}
Considering that the diagonal elements in the block $\bar{S}_2$ in
(\ref{eq:appU}) can be used to arbitrarily reduce entries in the
blocks above it (operations which correspond to unimodular row transformations), 
it will suffice to compute the right hand side of
(\ref{eq:appU2}) column-modulo $\bar{S}_2$.

\begin{lemma} \mylabel{lem:getrhs}
Given the blocks $E$, $Q$, $C$ and $K$ from Lemma~\ref{lem:getU12},
the matrix on the right of (\ref{eq:appU2}) can be computed
column-modulo $\bar{S}_2$ in $O((n+m)\bar{m}^{\omega-1} \M((\log
\det S)/\bar{m} + \log \bar{m}))$ bit operations.
\end{lemma}

\begin{proof}
The right hand side of (\ref{eq:appU2}) is equal to
\begin{equation} \mylabel{eq:ff}
\renewcommand{\arraystretch}{1.2}
\left [ \begin{array}{c} 
\bar{F}_2  \\
\phantom{\boldz}\\\hline
\bar{A}_2 \\\hline
\phantom{\boldz}  
\end{array} \right ]  +
\left [ \begin{array}{c} 
 Q  \\
I\\\hline
C \\\hline
K\end{array} \right ] E\bar{A}_2.
\end{equation} 
We will compute~(\ref{eq:ff}) column-modulo $\bar{S}_2$ in three steps.
We remark that, for readability, we have gathered
together all of the results that rely on partial linearization
in the next Section~\ref{sec:partial}. 
As such, we make forward references to Theorems~\ref{thm:partial}
and~\ref{thm:partial2} that establish the cost bounds for the matrix multiplications
column-modulo $\bar{S}_2$ that we require here in steps~1 and~2.

{\bf Step 1:} Compute $B = \colmod(E\bar{A}_2,\bar{S}_2)$. The row
dimension of $E$ and column dimension of $\bar{A}_2$ are bounded
by $\bar{m}$, and the inner dimension of the product $E \bar{A}_2$
is bounded by $n+m$. 
Since we want the result of $E \bar{A}_2$ column-modulo $\bar{S}_2$, we will
appeal to Theorem~\ref{thm:partial2} to do the product.  
In particular, from Lemma~\ref{lem:getU12} we have that 
$E = \rowmod(E,2sI)$, and, moreover, 
since $E$ has at most $\bar{m}$ rows, we have the bound
\begin{eqnarray}
\log \det 2sI  & \leq  & \bar{m}(\log 2s) \nonumber \\
 & =& \bar{m}\log 2 + \bar{m} \log s \nonumber \\
 & \leq &  \bar{m} \log 2 + 2 \log \det S \mylabel{eq:bnds},
\end{eqnarray}
where (\ref{eq:bnds}) follows from Lemma~\ref{lem:smallsnew}.  Let
$D$ be equal to the right hand side of (\ref{eq:bnds}).  By assumption,
$\bar{A}_2 = \colmod(\bar{A}_2,\bar{S}_2)$, and, moreover, since
$\bar{S}_2$ is a submatrix of $S$, we have the bound $\log \det
\bar{S}_2 \leq \log \det S \leq D$.  It now follows from
Theorem~\ref{thm:partial2}  that $B$ can be computed in the allotted
time.

{\bf Step 2:}  Compute
$$\colmod  \left (
\renewcommand{\arraystretch}{1.2}
\left [ \begin{array}{c} 
Q  \\
I \\\hline
C \\\hline
K
\end{array} \right ] B,\bar{S}_2  \right ).
$$
Both matrices in the product have column dimension bounded by
$\bar{m}$.  By Lemma~\ref{lem:getU12}, the first matrix is reduced
column modulo $2sI$, and by construction in step~2, the second
matrix $B$ is reduced column modulo $\bar{S}_2$.  By
Theorem~\ref{thm:partial}, the product can be done in the allotted
time with $D$ equal to the right hand side of (\ref{eq:bnds}).

{\bf Step 3:} Finally, add the result of step~2 to the matrix
containing $\bar{F}_2$ in (\ref{eq:ff}), and reduce the result
column modulo $\bar{S}_2$.  This can be done in the allotted time.
\end{proof}

\paragraph{Proof of Theorem~\ref{thm:hermbas}}
We first give the algorithm to transform the matrix in~(\ref{eq:shape2})
to its Hermite form.
To start stage~1, 
initialize the work matrix to be~(\ref{eq:shape2}),
with all $\bar{m} = m$ columns remaining to be transformed to correct form. 
In general, the work matrix has the shape shown
in~(\ref{eq:decomp1}), with $\bar{S}$ of dimension $\bar{m}$.
Decompose the work matrix as in (\ref{eq:shapek1}), with $\bar{S}_2$
of dimension $\lfloor \bar{m}/2 \rfloor$, and use Lemma~\ref{lem:getU12}
and \ref{lem:getrhs} 
to compute a unimodular transformation of the
work matrix so it has the shape shown on the right of~(\ref{eq:appU}).
This completes one stage.  Each stage 
puts the next $\lceil \bar{m}/2 \rceil$ columns into correct
form.  Because each stage is a unimodular row transformation (over
$\Z$) the resulting matrix at the end of all stages 
is the Hermite form of the starting matrix (\ref{eq:shape2}).

We now bound the cost of a particular stage in terms of $\bar{m}$.
By Lemma~\ref{lem:smallsnew}, the $\log s$ term in the cost estimate
of Lemma~\ref{lem:getU12} can be bounded from above by $2(\log \det
S)/\bar{m}$. By Lemmas~\ref{lem:getU12} and~\ref{lem:getrhs}, 
the cost of one stage is then 
\begin{equation} \mylabel{eq:needbnd} 
O((n+m)\bar{m}^{\omega-1}\,\B((\log \det S)/\bar{m} + \log \bar{m}))
\end{equation}
bit operations.  Let $D = \log \det S + m \log m$.  Then 
\begin{equation} \mylabel{eq:t1}
(\log \det S)/\bar{m} + \log \bar{m} \leq D/\bar{m}
\end{equation}
and
\begin{eqnarray} \bar{m}^{\omega-1} \,\B(D/\bar{m}) & =  &
m^{\omega-1} (\bar{m}/m)^{\omega-1} \, 
\B((m/\bar{m}) (D/m)) \nonumber \\
 & \leq & 
m^{\omega-1} (\bar{m}/m)^{\omega-1} \, \B(m/\bar{m}) 
\B(D/m) \nonumber \\
 &  \in  &  O(m^{\omega-1} \B(D/m)) \times (\bar{m}/m)^{\epsilon_0}, \mylabel{eq:t2}
\end{eqnarray}
where (\ref{eq:t2}) used $\B(m/\bar{m}) \in
O((m/\bar{m})^{\omega-1-\epsilon_0})$.  Each stage decreases $\bar{m}$
by at least half, so the sum of $(\bar{m}/m)^{\epsilon_0}$ over all
stages is a constant.  Substituting (\ref{eq:t1}) and (\ref{eq:t2})
into (\ref{eq:needbnd})  gives the result.

Now consider part~\ref{thm:hb2} of the theorem.
Consider running a simplified version of the algorithm just described
when $T$ is already known: the input
matrix is now as in~(\ref{eq:shape2}) but with the block $\boldz$ replaced by $T$.  
In this case equation~(\ref{eq:lemcompa}) changes as follows:
the computation of $U_1$ is entirely avoided (ie, $U_1$ is the identity)
and the blocks $Q$ and $\ast_2$ in $U_2$ are zero since $\bar{F}_1=\bar{G}_1$.

Considering the shape of $U_2$, if we apply all the unimodular transformations
computed in the simplified algorithm just described to the input matrix
\begin{equation} \mylabel{eq:inxx}
\left [ \begin{array}{ccc}  T & & I_m \\
A & I_n & \\
S & &   \end{array} \right ]
\end{equation}
we will transform it (via unimodular row transformations over $\Z$) to the matrix
\begin{equation} \mylabel{eq:outxx}
\left [ \begin{array}{ccc}  T & & I_m \\
  & I_n & C  \\
  & & K \end{array} \right ]
\end{equation}
which has trailing $(n+m) \times (n+m)$ submatrix in Hermite form.  
\qed

\begin{example} As in Example~\ref{ex:hermbas}, suppose our input matrices 
have sizes $n=3$ and $m=7$, that is, $A \in \Z^{3 \times 7}$
and $S \in \Z^{7 \times 7}$. But now we also have the Hermite basis $T \in \Z^{n \times n}$. 
We illustrate the transform of~(\ref{eq:inxx}) to~(\ref{eq:outxx}), focusing
on the last block-column of~(\ref{eq:inxx}).
The process becomes
\begin{equation}\mylabel{eq:exiter}
\stackrel{1}{\rightarrow} 
\setlength{\arraycolsep}{.55\arraycolsep} \renewcommand{\arraystretch}{.6}
\left [ \begin{array}{cccc|ccc}
  1   &        &        &        &        &         &        \\
       &   1    &        &        &        &         &        \\
       &        &   1    &        &        &         &        \\
       &        &        &   1    &        &         &        \\
       &        &        &        &   1    &         &        \\
       &        &        &        &        &    1    &        \\
       &        &        &        &        &         &       1\\\hline
 \ast  &   \ast &   \ast &   \ast &   \phantom{\ast}&  \phantom{\ast}  & \phantom{\ast}   \\
 \ast  &   \ast &   \ast &   \ast &   \phantom{\ast}&   \phantom{\ast} & \phantom{\ast}   \\
 \ast  &   \ast &   \ast &   \ast &   \phantom{\ast} &   \phantom{\ast}  & \phantom{\ast}   \\\cline{1-4}
 \underline{\ast} &   \ast &   \ast &   \ast &   \phantom{\ast}&    \phantom{\ast}  &  \phantom{\ast}  \\
       & \underline{\ast}  &  \ast  &  \ast  &   \phantom{\ast}&   \phantom{\ast}  &   \phantom{\ast}  \\
       &        & \underline{\ast}   & \ast   &   \phantom{\ast}   &  \phantom{\ast}   &   \phantom{\ast} \\
       &        &        & \underline{\ast}   &  \phantom{\ast}&  \phantom{\ast}  &  \phantom{\ast}  \\\cline{1-4}
       &        &        &        &   \phantom{\ast}  &         &        \\
       &        &        &        &        &   \phantom{\ast}   &        \\
       &        &        &        &        &         &   \phantom{\ast}  
\end{array} \right ]  \stackrel{2}{\rightarrow}
\left [ \begin{array}{cccc|cc|c}
  1   &        &        &        &        &         &        \\
       &   1    &        &        &        &         &        \\
       &        &   1    &        &        &         &        \\
       &        &        &   1    &        &         &        \\
       &        &        &        &   1    &         &        \\
       &        &        &        &        &    1    &        \\
       &        &        &        &        &         &       1\\\hline
 \ast  &   \ast &   \ast &   \ast &   {\ast}&  {\ast}  & \phantom{\ast}   \\
 \ast  &   \ast &   \ast &   \ast &  {\ast}&   {\ast} & \phantom{\ast}   \\
 \ast  &   \ast &   \ast &   \ast &   {\ast} &   {\ast}  & \phantom{\ast}   \\\cline{1-4}
 \underline{\ast} &   \ast &   \ast &   \ast &   {\ast}&    {\ast}  &  \phantom{\ast}  \\
       & \underline{\ast}  &  \ast  &  \ast  &   {\ast}&  {\ast}  &   \phantom{\ast}  \\
       &        & \underline{\ast}   & \ast   &   {\ast}   &  {\ast}   &   \phantom{\ast} \\
       &        &        & \underline{\ast}   &  {\ast}&  {\ast}  &  \phantom{\ast}  \\\cline{1-6}
       &        &        &        &   \underline{\ast}  & \ast        &        \\
       &        &        &        &        &   \underline{\ast}   &        \\ \cline{5-6}
       &        &        &        &        &         &   \phantom{\ast}
\end{array} \right ]  \stackrel{3}{\rightarrow}
\left [ \begin{array}{cccc|cc|c}
  1   &        &        &        &        &         &        \\
       &   1    &        &        &        &         &        \\
       &        &   1    &        &        &         &        \\
       &        &        &   1    &        &         &        \\
       &        &        &        &   1    &         &        \\
       &        &        &        &        &    1    &        \\
       &        &        &        &        &         &       1\\\hline
 \ast  &   \ast &   \ast &   \ast &   {\ast}&  {\ast}  & {\ast}   \\
 \ast  &   \ast &   \ast &   \ast &  {\ast}&   {\ast} & {\ast}   \\
 \ast  &   \ast &   \ast &   \ast &   {\ast} &   {\ast}  & {\ast}   \\\cline{1-4}
 \underline{\ast} &   \ast &   \ast &   \ast &   {\ast}&    {\ast}  &  {\ast}  \\
       & \underline{\ast}  &  \ast  &  \ast  &   {\ast}&  {\ast}  &   {\ast}  \\
       &        & \underline{\ast}   & \ast   &   {\ast}   &  {\ast}   &   {\ast} \\
       &        &        & \underline{\ast}   &  {\ast}&  {\ast}  &  {\ast}  \\\cline{1-6}
       &        &        &        &   \underline{\ast}  & \ast        &      \ast  \\
       &        &        &        &        &   \underline{\ast}   &     \ast   \\ \cline{5-7}
       &        &        &        &        &         &   \underline{\ast}
\end{array} \right ] = 
\left [ \begin{array}{ccccccc}
  1   &        &        &        &        &         &        \\
       &   1    &        &        &        &         &        \\
       &        &   1    &        &        &         &        \\
       &        &        &   1    &        &         &        \\
       &        &        &        &   1    &         &        \\
       &        &        &        &        &    1    &        \\
       &        &        &        &        &         &       1\\\hline
 \ast  &   \ast &   \ast &   \ast &   {\ast}&  {\ast}  & {\ast}   \\
 \ast  &   \ast &   \ast &   \ast &  {\ast}&   {\ast} & {\ast}   \\
 \ast  &   \ast &   \ast &   \ast &   {\ast} &   {\ast}  & {\ast}   \\\hline
 \underline{\ast} &   \ast &   \ast &   \ast &   {\ast}&    {\ast}  &  {\ast}  \\
       & \underline{\ast}  &  \ast  &  \ast  &   {\ast}&  {\ast}  &   {\ast}  \\
       &        & \underline{\ast}   & \ast   &   {\ast}   &  {\ast}   &   {\ast} \\
       &        &        & \underline{\ast}   &  {\ast}&  {\ast}  &  {\ast}  \\
       &        &        &        &   \underline{\ast}  & \ast        &      \ast  \\
       &        &        &        &        &   \underline{\ast}   &     \ast   \\ 
       &        &        &        &        &         &   \underline{\ast}
\end{array} \right ].
\end{equation}
Stage~1 starts by filling in the first four columns of $C$ and $K$ below the first four diagonals.
Stage~2 deposits the next two columns of $C$ and $K$ below the fourth and fifth diagonals,
leaving the first four columns unchanged. Stage~3 fills in the last column and completes.
\end{example}


\begin{example}
Let $A = \left[ \begin{array}{cccc} 1 & 5 & 19  \end{array} \right]$ and $S = \diag(2,6,  72)$. Then 
$$
T = \left[ \begin{array}{cccc} 1 & 1 & 5 \\   & 2 & 4  \\   &   & 6  \end{array} \right]
\mbox{~~is the Hermite form of~~} \left[ \begin{array}{c}  A \\ S \end{array} \right ].$$
We will show how the algorithm transforms
$$ 
\left [ \begin{array}{c|c} T & I_3 \\
A & \\
S &  \end{array} \right ] \stackrel{1}{\rightarrow} \cdot \stackrel{2}{\rightarrow} \left [ \begin{array}{c|c} T & I_3 \\
       & C \\
       & K  \end{array} \right ] 
$$ 
to produce $C$ and $K$.  Transforming the first two columns in stage~1 and then the last column in stage~2, 
we obtain
$$
\left[ \begin{array}{ccc||ccc} 1 & 1 & 5  & 1 \\   & 2 & 4  & & 1 \\   &   & 6& & & 1 \\  \hline 1 & 5 & 19 & \\ \cline{1-3}  2 & & \\ & 6 &  \\ & & 72  \end{array} \right ] \stackrel{1}{\rightarrow} 
\left[ \begin{array}{ccc||ccc} 1 & 1 & 5  & 1 \\   & 2 & 4  & & 1 \\   &   & 6& & & 1 \\  \hline  &  & 24 & 1 &0  & \\  & & 18 & 2 & 2 &  \\ & & 18& & 3 &  \\ \cline{3-5} & & 72  \end{array} \right ] \stackrel{2}{\rightarrow} 
 \left[ \begin{array}{ccc||ccc} 1 & 1 & 5  & 1 \\   & 2 & 4  & & 1 \\   &   & 6& & & 1 \\  \hline  &  &  &  1 & 0& 8 \\ \cline{4-6}  &  &  &  2 & 2 & 9 \\ & & & & 3 & 10 \\ & & & & & 12  \end{array} \right ]
$$
giving 
$$
C := \left[ \begin{array}{ccc} 1 & 0 & 8 \end{array} \right] ~~\mbox{ and } ~~ K := \left [ \begin{array}{ccc} 2 & 2 & 9 \\ & 3 & 10 \\ & & 12 \end{array} \right].
$$
We remark that the first stage computes modulo~$6^2$
while the second stage computes modulo~$72^2$.  The computation
that produces column~3 in the second matrix uses formula~(\ref{eq:appU2}),
but with $U_1$ the identity matrix, and blocks $Q$ and $\ast_2$ in $U_2$ the zero matrix.
\end{example}

\section{Matrix multiplication modulo diagonal matrices}\mylabel{sec:partial}

Algorithm \texttt{HermiteBasis} makes use of matrix multiplication
modulo a diagonal matrix in the operations $F_1 := \colmod(FV_1,S_1)$,
$F_2 := \colmod(CV_2,S_2)$ and $B := \colmod(H_1F,S)$.
The first two of these have the form
``tall $\times$ square,'' while the last is of the form ``Hermite $\times$ tall.''
The Hermite basis computation in Section~\ref{sec:hermbas}
requires an additional type of the form ``wide $\times$ tall'' 
for the operation $B = \colmod(E\bar{A}_2,\bar{S}_2)$.
In this section we show how to do these operations with a good bit complexity
rather than algebraic complexity.  Our constructions are based on partial linearization.
We begin by stating the results we need.
The actual construction and proofs of Theorem~\ref{thm:partial} and~\ref{thm:partial2}
are in Section~\ref{subsec:partial}.  

Consider first computing the product 
of an $n \times m$ matrix $A$ and $m \times m$ matrix $B$, with $n \geq m$:
$$
\left [ \begin{array}{c} \phantom{A} \\
A \\
\phantom{A} \end{array} \right ]
\left [  \begin{array}{c} B\end{array} \right ]=
\left [ \begin{array}{c} \phantom{A} \\
C \\
\phantom{A} \end{array} \right ].
$$
In an algebraic model, $C$ is computed in $O(nm^{\omega-1})$
arithmetic operations from the domain of entries.  Here  we consider
the cost in bit operations of computing $AB$ over the integers, but
in the following scenario where we are given  nonsingular diagonal
matrices $E$ and $F$:
\begin{itemize}
\item[(i)]  We know that $A=\colmod(A,E)$ and $B=\colmod(B,F)$.
\item[(ii)]  We only need $\colmod(AB,F)$.
\end{itemize}

Let $D$ be a common upper bound for both $\log \det E$ and $\log \det F$.
Then $D$ is an upper bound for the sum of the logarithms of
the $m$ diagonal entries of both $E$ and $F$, and the
average bitlength of the columns of $A$ and $B$ is bounded by $O(D/m +1)$.
Up to log factors, Theorem~\ref{thm:partial} shows that $\colmod(AB,F)$
can be computed in a number of bit operations equal to the algebraic cost 
$O(nm^{\omega-1})$ times $D/m+1$.

\begin{theorem} \mylabel{thm:partial} 
Let $E,F \in \Z_{\geq 0}^{m \times m}$ be nonsingular diagonal matrices.
Let $A =\colmod(A,E) \in \Z^{n \times m}$ and $B= \colmod(B,F) \in \Z^{m \times m}$,
with $n \geq m$. 
Then $\colmod(AB,F)$ 
can be computed in 
$O(nm^{\omega-1}\, \M(D/m + \log m))$ 
bit operations,
where $D$ is an upper bound 
for $\log \det E$ and $\log \det F$.
\end{theorem}

We will also need to apply the column-modulo matrix multiplication
technique of Theorem~\ref{thm:partial} with
an input matrix $A$ that may not necessarily have all entries
nonnegative.  To this end, let $A^{(+)}$ denote the the matrix $A$
but with all negative entries zeroed out.
Then $A = A^{(+)} - (-A)^{(+)}$ with $A^{(+)}$ and $(-A)^{(+)}$ being over $\Z_{\geq 0}$.

\begin{corollary} \mylabel{cor:partial1}  
Let $F \in \Z_{\geq 0}^{m \times m}$ be a nonsingular diagonal matrix.
Let $A \in \Z^{n \times m}$ and $B= \colmod(B,F) \in \Z^{m \times m}$,
with $n \geq m$. 
Then $\colmod(AB,F)$ 
can be computed in 
$O(nm^{\omega-1}\, \M(D/m + \log m))$ 
bit operations,
where $D$ is an upper bound for both $\log \det F$ and
the sum of the bitlengths of the columns of $A$.
\end{corollary}
\begin{proof}
Let $d_i$ be the bitlength of column $i$ of $A \in \Z^{n \times m}$, $1\leq i\leq m$.
Then $A^{(+)}$ and $(-A)^{(+)}$ are reduced column-modulo $E$,
where $E=\diag(2^{d_1},2^{d_2},\ldots,2^{d_m})$. By assumption, 
$D \geq \log_2 \det E$.  Use Theorem~\ref{thm:partial}
to compute $C_1 := \colmod(A^{(+)}B,F)$ and $C_2 := \colmod((-A)^{(+)}B,F)$
in the allotted time, and then compute $\colmod(AB,F)$ as $\colmod(C_1
- C_2,F)$.
\end{proof}

\begin{corollary}  \mylabel{cor:partial}
Suppose we are given as input
\begin{enumerate}
\item $H \in \Z^{n \times n}$, a nonsingular Hermite basis
with at most $m$ nontrivial columns,
\item $S \in \Z^{m \times m}$, a nonsingular Smith form with $m \leq n$, and
\item $M \in \Z^{n \times m}$, satisfying $M = \colmod(M,S)$.
\end{enumerate}
Then $\colmod(HM,S)$ can be computed in
$O(nm^{\omega-1}\, \M(D/m + \log m))$ 
bit operations,
where $D$ is an upper bound for $\log \det H$ and  $\log \det S$.
\end{corollary}
\begin{proof}
Since $HM = (H-I)M + M$, it will suffice to show how to
compute $\colmod((H-I)M,S)$ in the allotted time. By
assumption, $H-I$ has at least $n-m$ zero columns. Let $\mathcal S
\subseteq \{1,2,\ldots, n\}$ be a subset of column indices, with
$|\mathcal S|=m$, which captures all the nonzero columns of $H-I$,
and let $\bar{H} \in \Z^{n\times m}$ be the submatrix of $H-I$ comprised
of columns $\mathcal S$. Let $\bar{M} \in \Z^{m \times m}$ be the
subset of $M$ comprised of rows $\mathcal S$.
Then $\colmod((H-I)M,S) = \colmod(\bar{H}\bar{M},S)$,
which by Theorem~\ref{thm:partial} can be computed in the stated time.
\end{proof}

Theorem~\ref{thm:partial2} gives the analogous result of Theorem~\ref{thm:partial} in the case where  $A$ is transposed:
$$
\left [ \begin{array}{ccc} \phantom{A} &
A & \phantom{A} \end{array} \right ]
\left [ \begin{array}{c} \phantom{A} \\
B \\
\phantom{A} \end{array} \right ] =
\left [  \begin{array}{c} C\end{array} \right ].
$$

\begin{theorem} \mylabel{thm:partial2} 
Let $E,F \in \Z_{\geq 0}^{m \times m}$ be nonsingular diagonal
matrices.  Let $A =\rowmod(A,E) \in \Z^{m \times n}$ and $B=
\colmod(B,F) \in \Z^{n \times m}$, with $n \geq m$.  Then $\colmod(AB,F)$ 
can be computed in 
$O(nm^{\omega-1}\, \M(D/m + \log m))$ 
bit operations, where $D$ is upper bound for $\log
\det E$ and $\log \det F$.
\end{theorem}
 
\subsection{Proofs of Theorems~\ref{thm:partial} and \ref{thm:partial2}}\mylabel{subsec:partial}

\paragraph{Partial linearization}
We will make use of partial linearization \cite{StorjohannDagstuhl06}
in the proofs of Theorems~\ref{thm:partial} and \ref{thm:partial2}
and so begin by defining some notation used for this construction.
For a modulus  $X \in \Z_{>0}$ and length $t \in \Z_{\geq 0}$,
define $\vec{X}^{(t)}$ to be  the expansion/compression vector
$$
\vec{X}^{(t)} := \left [ \begin{array}{c} 1 \\
X \\
\vdots \\
X^{t-1} 
\end{array} \right ] \in \Z^{t \times 1}.
$$
Then premultiplying $u \in \Z^{1 \times \ast}$ expands one row into $t$ rows:
$$
\vec{X}^{(t)} u = \left [ \begin{array}{c} u \\
Xu \\
\vdots \\
X^{t-1}u \end{array} \right ] \in \Z^{t \times \ast}.
$$
Postmultiplying $v = \left [ \begin{array}{cccc} v_0 & v_1 &
\cdots & v_{t-1}\end{array} \right ] \in \Z^{\ast \times t}$
compresses $t$ columns into one column:
$$
v \vec{X}^{(t)}  = \left [ \begin{array}{c} v_0 +v_1X + \cdots + 
v_{t-1} X^{t-1} \end{array} \right ] \in \Z^{\ast \times 1}.
$$

For a tuple $\vec{t} = (t_1,t_2,\ldots,t_m)$ of nonnegative integers, define
$$\vec{X}^{(\vec{t})}:=
\left [ \begin{array}{cccc} \vec{X}^{(t_1)} \\
 & \vec{X}^{(t_2)} \\
 & & \ddots \\
 & & & \vec{X}^{(t_m)} \end{array} \right ] \in \Z^{|\vec{t}| \times m},
$$
where $|\vec{t} |= \sum_i t_i$.

\paragraph{Proof of Theorem  \ref{thm:partial}}
Let $X\in \Z_{>0}$ be the smallest power of two such that that $\log
X \geq D/m$.  We will construct $\bar{A} \in [0,X)^{n \times
|\vec{e}|}$ with $|\vec{e}| < 2m$, and $\bar{B} \in [0,X)^{|\vec{e}|
\times |\vec{f}|}$ with  $|\vec{f}| < 2m$,  such that \begin{equation}
\mylabel{eq:target} \colmod(AB,F) =
\colmod(\bar{A}\bar{B}\vec{X}^{(\vec{f})},F).  \end{equation}

Let $\bar{A}$ be the matrix obtained from $A$ by replacing column
$i$ ($1\leq  i \leq m$) with the $n \times e_i$ matrix corresponding
to the coefficients of its $X$-adic expansion.   Note that we can
take $e_i\geq 0$ to be minimal such that $X^{e_i}\geq E_{ii}$,
giving the upper bound
\begin{eqnarray*} e_i & < & 1+(\log E_{ii})/(\log X) \\
 & \leq & 1 + (\log E_{ii})/(D/m) \\
 & = & 1 + (m/D)\log E_{ii}.
\end{eqnarray*}
Since $D \geq \log \det E$, we see that the column dimension 
$|\vec{e}|=\sum_{i=1}^m e_i$ of $\bar{A}$ satisfies $|\vec{e}| < 2m$.

Note that $A = \bar{A}  \vec{X}^{(\vec{e})}$
and thus $AB = \bar{A} \vec{X}^{(\vec{e})}B$.  
Define $\hat{B} = \colmod( \vec{X}^{(\vec{e})}B,F) \in \Z^{|\vec{e}| \times m}$.  
Then $\colmod(AB,F) = \colmod(\bar{A} \hat{B},F)$.
Let $\bar{B}$ be the matrix obtained from $\hat{B}$ by replacing
column $i$ ($1\leq i\leq m$) with the $e \times f_i$ matrix
corresponding to its $X$-adic expansion. Similar to $\bar{A}$,
we can take $f_i\geq 0$ to be minimal such that $X^{f_i}\geq F_{ii}$,
in which case the column dimension $|\vec{f}|$ of $\bar{B}$
satisfies $|\vec{f}| < 2m$.  At this point we have constructed
$\bar{A}$ and $\bar{B}$ such that (\ref{eq:target}) is satisfied.

The complete recipe to compute $C=\colmod(AB,F)$ using partial linearization is as follows:

\begin{tabbing}
\hbox{~~~~~~~~~~~~~~~} \= step 1 : Construct $\bar{A}$ such that 
           $A = \bar{A} \vec{X}^{(\vec{e})}$ \hbox{~~~} \= : column linearization \\
\>step 2 :  Compute $\hat{B} := \colmod( \vec{X}^{(\vec{e})}B,F)$ \> : row expansion \\
\> step 3 : Construct $\bar{B}$ such that $\hat{B} = \bar{B} \vec{X}^{(\vec{f})}$ \>
 : column linearization  \\
\> step 4 : $\bar{C} := \bar{A}\bar{B}$ \> : product over $\Z$  \\
\> step 5 : $C := \colmod(\bar{C}\vec{X}^{(\vec{f})},F)$ \> : column compression 
\end{tabbing}

We now give a cost analysis in terms of bit operations. 
The ``product over $\Z$'' step~4 dominates the cost so we consider it first. 
Since both $\bar{A}$ and
$\bar{B}$ have column dimension $<2m$, and $\bar{A}$ has row
dimension $n\geq m$, the cost is $O(n m^{\omega-1} \, \M(\log X +
\log m))$.  (The $\log m$ additive term in $\M(\log X + \log m)$
accounts for the fact that entries of $\bar{C} = \bar{A}\bar{B}$ can have
magnitude $O(mX)$ since they are the result of inner products of
length $<2m$.)  Since $\log X \in O(D/m)$, the cost is as stated
in the theorem. Next we show that steps~1, 2, 3 and~5 can be done 
in a time that is dominated by the cost bound just derived for step~4.

Since $X$ is a power of~2, the two ``column linearization'' steps~1 and~3
can be done in linear time $O(nm\log X)$.

The ``row expansion'' step~2 can be done
using $|\vec{e}|-m \leq m$ iterations of the following operation: multiply by $X$
a row vector in $\Z^{1 \times m}$ that is reduced column-modulo
$F$, and then reduce it column-modulo $F$.  The cost is $O(m
\sum_{i=1}^m \M(\log X + \log F_{ii}))$, which simplifies to $O(m
\, \M(m (\log X + D/m)))$ using first the superlinearity of $\M$
and then the assumption that $D \geq \log \det F$.  Simplifying the
$\M(m(\log X + D/m))$ term in this cost estimate using $\M(m(\log
X + D/m)) \leq \M(m) \M(\log X + D/m)$ and $\M(m) \in
O(m^{\omega-1-\epsilon_0})$ shows that the cost of step~2 is bounded
by $O(m^{\omega-\epsilon_0}\,\M(D/m + \log X))$ bit operations.

For the ``column compression'' step~5 we first compute $\bar{C}
\vec{X}^{(\vec{f})}$.  Since $||\bar{C}|| \in O(m X)$ 
and $X$ is a power of~2, the binary representation of
entries in $\bar{C}\vec{X}^{(\vec{f})}$ can be computed, in order from least to most
significant bit, with $n(|\vec{f}|-m)  < nm$ additions of integers having bitlength $O(\log
X + \log m)$, that is, in $O(nm(\log X + \log m))$ bit operations.
It remains to
reduce entries in $\bar{C} \vec{X}^{(\vec{f})}$ column-modulo $F$.
A similar analysis as done for the ``row expansion''  step~2, now
allowing for the fact that column $i$ of $\bar{C} \vec{X}^{(\vec{f})}$
can have bitlength $O(\log F_{ii} + \log m)$, shows that the cost
of this column-modulo $F$ operation is bounded by
$O(nm^{\omega-1-\epsilon_0}\,\M(D/m + \log X + \log m))$ bit operations. \qed

\paragraph{Proof of Theorem  \ref{thm:partial2}}

Let $X\in \Z_{>0}$ be the smallest power of two such that that $\log
X \geq D/m$.  We are going to construct an $\bar{A} \in [0,X)^{|\vec{e}|
\times n}$ with $|\vec{e}| < 2m$, and $\bar{B} \in [0,X)^{n
\times |\vec{f}|}$ with $|\vec{f}| < 2m$, with

\begin{equation} \mylabel{eq:target2} 
\colmod(AB,F) = \colmod\bigl (\hbox{Transpose}(\vec{X}^{(\vec{e})})\bar{A}\bar{B}\vec{X}^{(\vec{f})},F\bigr ).
\end{equation}

Let $\bar{A}$ be the matrix obtained from $A$ by replacing row $i$
($1\leq  i \leq m$) with the $e_i \times n$ matrix corresponding
to the coefficients of its $X$-adic expansion.  Let $\bar{B}$ be
the matrix obtained from $B$ by replacing column $i$ ($1\leq i\leq
m$) with the $n \times f_i$ matrix corresponding to its $X$-adic
expansion.  If we choose $e_i$ and $f_i$ as in the proof of
Theorem~\ref{thm:partial}, we have $|\vec{e}|,|\vec{f}| < 2m$.
At this point we have constructed $\bar{A}$ and $\bar{B}$ such that
(\ref{eq:target2}) holds.

Decompose $\bar{A}$ and $\bar{B}$ conformally into $t = \lceil n/m\rceil $ blocks so that
\begin{equation} \mylabel{eq:decomp}
\bar{A}\bar{B} = \stackrel{\textstyle \bar{A}}{\left [ \begin{array}{c|c|c|c} \bar{A}_1 & \bar{A}_1 & \cdots
 & \bar{A}_t \end{array} \right ]}
\stackrel{\textstyle \bar{B}}{
\renewcommand{\arraystretch}{1.2}
\left [ \begin{array}{c} \bar{B}_1 \\\hline
\bar{B}_2 \\\hline
\vdots \\\hline
\bar{B}_t
\end{array} \right ]},
\end{equation}
with each block $\bar{A}_{\ast}$ of column dimension $m$, except
for possibly $\bar{A}_t$, which may have fewer than $m$ columns.
Then
\begin{eqnarray*}
\colmod(AB,F) & = & \colmod\bigl(\sum_{i=1}^t \hbox{Transpose}(\vec{X}^{(\vec{e})})\bar{A}_i\bar{B}_i\vec{X}^{(\vec{f})},F
\bigr ) \\
   & = & \colmod\bigl(\hbox{Transpose}(\vec{X}^{(\vec{e})})
\colmod\bigl ( \sum_{i=1}^t \bar{A}_i\bar{B}_i\vec{X}^{(\vec{f})},F\bigr),F\bigr).
\end{eqnarray*}

The complete recipe to compute $C=\colmod(AB,F)$ using partial linearization is as follows:
\begin{tabbing}
\hbox{~~~~~~~} \= step 1 : Construct $\bar{A}$ such that $A = \vec{X}^{(\vec{e})}\bar{A}$ \hbox{~~~} \= : row linearization \\
\> step 2 : Construct $\bar{B}$ such that $B = \bar{B} \vec{X}^{(\vec{f})}$ \> : column linearization  \\
\> step 3 : Decompose $\bar{A}$ and $\bar{B}$ as in (\ref{eq:decomp}) \> : block decomposition \\
\> step 4 : $\bar{C}_i := \bar{A}_i\bar{B}_i$, $1\leq i\leq t$ \> : products over $\Z$  \\
\>step 5 : $\hat{C}_i := \colmod\bigl(\bar{C}_i\vec{X}^{(\vec{f})},F\bigr )$, $1\leq i\leq t$ \> : column compressions \\
\> step 6 : $\hat{C} := \colmod\bigl( \sum_{i=1}^t \hat{C}_i,F\bigr)$ \> : modular matrix sum \\
\> step 7 : $C := \colmod(\hbox{Transpose}(\vec{X}^{(\vec{e})})\hat{C},F)$ \> : row compression
\end{tabbing}

We now give a cost analysis in terms of bit operations.  Similar
to the proof of Theorem~\ref{thm:partial}, the ``products over
$\Z$'' step~4 has cost $O(n m^{\omega-1} \, \M(\log X + \log m))$.
Next we show that the other steps are dominated by this cost.

Like in the the proof of Theorem~\ref{thm:partial}, the ``row
linearization'' step~1 and ``column linearization'' step~2 can both be done in
linear time $O(nm\log X)$. The ``block decomposition'' step~3 does not involve
any computation.  Similar to the proof of Theorem~\ref{thm:partial},
the ``column compressions'' step~5 can be done in time
$O(nm^{\omega-1-\epsilon_0}\,\M(D/m + \log X + \log m)$, which dominates
the cost of the ``modular matrix sum'' step~6.

Finally, consider the ``row compression'' step~7.  Using a procedure
with computational steps identical to those in the recipe for the
 ``row expansion'' step in the proof of Theorem~\ref{thm:partial}, but incorporating
here a Horner type scheme 
to accumulate the results using some additional summations,
gives the same cost bound of $O(m^{\omega-\epsilon_0}\, \M(D/m +\log X))$
as derived in Theorem~\ref{thm:partial}. \qed

\section{Complexity analysis of Algorithm~\texttt{HermiteBasis}} \mylabel{sec:runtime}

Algorithm~\texttt{HermiteBasis}$(S,F,k,m)$ takes as input a reduced Smith
massager $(S,F)$ that has an index $(k,m)$ Hermite basis.
By Corollary~\ref{cor:atmostm}, such an $S$ will have at most $m$
nontrivial invariant factors, and so
by Lemma~\ref{lem:sreduce} 
the additional precondition of the algorithm 
that $S$ and $F$ have column dimension $m$ is not an essential restriction.
In this section we establish the following result.

\begin{theorem} \mylabel{thm:hbasrun} If, up to trivial invariant
factors in $S$ and corresponding leading zero columns in $F$,
$[\epsilon](S,F,k,m)$ is a valid input to Algorithm~\texttt{HermiteBasis},
then the Hermite basis of $\Rel(S,F)$ can be computed in\\
$
O(nm^{\omega-1}\, \B((\log \det S)/m + \log m)(\log m)^2\log(1/\epsilon))
$
bit operations.  The algorithm is randomized of the Las Vegas
type: it either returns the correct output, or reports \textsc{fail}
with probability at most $\epsilon$.
\end{theorem}

We defer the proof of Theorem~\ref{thm:hbasrun} to the end of this
section.

For convenience, we will initially assume that $m$ is a power of
two.  The recursion tree for Algorithm~\texttt{HermiteBasis} is
then a perfect binary tree of height $\log_2 m$.  Lemma~\ref{lem:rundet}
bounds the total cost of all work except for the calls to
the \texttt{SmithMassager}.  Lemma~\ref{lem:runprob} then bounds
the total cost of all calls to the randomized
algorithm~\texttt{SmithMassager}. Lemma~\ref{lem:ext} shows
that assuming $m$ to be a power of two is not an essential restriction, 
with the proof of Theorem~\ref{thm:hbasrun} then following.

\begin{lemma} \mylabel{lem:rundet} 
Not counting the cost of the calls to \texttt{SmithMassager},
Algorithm \texttt{HermiteBasis} has running time bounded by
$O(nm^{\omega-1}\, \B((\log \det S)/m + \log m))$
bit operations. This result assumes that $m$ is a power of two.
\end{lemma}
\begin{proof}
We will prove that the sum of the cost of the nonrecursive work
over all nodes in the recursion tree, not counting the cost of the
calls to \texttt{SmithMassager}, satisfies the bound stated in the
theorem.  For brevity, let $D := \log \det S$.  When discussing a
particular node, we will distinguish the current value of a parameter
by using a bar, for example, a particular node corresponds to the
calling sequence $(\bar{S},\bar{F},\ast,\bar{m})$.

The internal nodes correspond to calling sequences with $\bar{m} >1$.
Each internal node has two children corresponding to the two recursive
calls. The cost of the nonrecursive work at a particular internal
node depends on $n$ (fixed for all nodes), $\bar{m}$ (decreasing
by a factor of~2 for the subproblems) and $\bar{D} := \log \det
\bar{S}$ (splitting into $\bar{D} = \bar{D}_1 + \bar{D}_2$ corresponding
to $\det \bar{S} = (\det \bar{S}_1)(\det \bar{S}_2)$ for the
subproblems).  Level $i$ of the tree has $2^i$ nodes, each with
$\bar{m}=m/2^i$.  The sum of $\bar{D}$ over all nodes at any
particular level is equal to $D$.

Consider the top level of the recursion in Algorithm~\texttt{HermiteBasis}.
We claim that, ignoring the two calls to \texttt{SmithMassager}, and up to 
a multiplicative constant, the
cost of the nonrecursive work in the ``\texttt{else}'' case is
bounded  by
\begin{equation} \mylabel{eq:tcmulcost1}
n\bar{m}^{\omega -1} \, \B(\bar{D}/\bar{m}
+ \log_2 \bar{m}).
\end{equation}
Theorem~\ref{thm:hermbas} gives the bound for computing $T$ and
$(K,C)$.  Theorem~\ref{thm:partial} gives the bound
for computing $\colmod(FV_1,S_1)$ and $\colmod(CV_2,S_2)$.
Corollary~\ref{cor:partial} gives the bound for computing
$\colmod(H_1F,S)$. (In fact, for these three partially linearized
multiplications the slightly better bound~(\ref{eq:tcmulcost1}) with $\B$ replaced
by $\M$ also holds.)

At this point we have established that each internal node of the
recursion tree has associated cost bounded by~(\ref{eq:tcmulcost1}).
We now bound the sum of~(\ref{eq:tcmulcost1}) over all nodes in a
non-leaf level $i$.  Any non-leaf level has $\bar{m} \geq 2$ and
thus the argument of $\B$ in~(\ref{eq:tcmulcost1}) satisfies
$\bar{D}/\bar{m} + \log_2 m\geq 1$.
The superlinearity of $\B$ gives that
the sum of~(\ref{eq:tcmulcost1}) over all nodes at a non-leaf level
$i$ is bounded by
\begin{eqnarray}
n\bar{m}^{\omega-1}\, \B(D/\bar{m} + 2^i\log_2 \bar{m}) & = & n(m/2^i)^{\omega-1}\, \B(D/(m/2^i) + 2^i\log_2 \bar{m}) \nonumber \\
 &  = & nm^{\omega-1}(1/2^i)^{\omega-1} \,
\B( 2^i\times (D/m +\log_2 \bar{m})) \nonumber \\
 & \leq & nm^{\omega-1}(1/2^i)^{\omega-1} \, \B(2^i) \B(D/m + \log_2 \bar{m}) \nonumber\\
 & \in & O(nm^{\omega-1} \B(D/m + \log \bar{m})) \times \frac{1}{(2^{\epsilon_0})^{i}}.
\mylabel{eq:bsecond}
\end{eqnarray}
Here, we used $\B(2^i) \in O((2^i)^{\omega-1-\epsilon_0})$
to obtain (\ref{eq:bsecond}).

If we substitute the global upper bound  $\log m$ for $\log \bar{m}$
in~(\ref{eq:bsecond}), then the only term in (\ref{eq:bsecond})
that depends on $i$ is $1/(2^{\epsilon_0})^i$.  Summing
$(1/2^{\epsilon_0})^i$ over all levels $i\geq 0$ yields a constant,
so the total cost at all internal nodes is $O(nm^{\omega-1}\,
\B(D/m+\log m))$, which is bounded by the cost stated in the theorem.

It remains to bound the cost of the base cases 
corresponding to $\bar{m}=1$.
The base cases coincide with calling sequences 
$(h_iI_1,\ast,k+i-1,1)$ for $1\leq i\leq m$,
with the $h_i$ being the
diagonal entries of the block $\bar{H}$ in the index $(k,m)$ Hermite basis $H$ of $\Rel(S,F)$,
see~(\ref{eq:partm}).
Lemma~\ref{lem:tbasecase} bounds the cost of a base case
by $O(n\, \B((\log h_i)+1))$, where we have added $+1$ to ensure
the argument of $\B$ is bounded from below by 1.
The superlinearity of $\B$ then
gives
\begin{eqnarray*}
\sum_{i=1}^m n\, \B((\log h_i)+1) &\leq &n\, \B(D+m)\\ 
    &= & n \, \B(m ((D/m)+1)) \\
 & \leq & n\, \B(m)\B(D/m+1)  \\
&  \in & O(nm^{\omega-1-\epsilon_0}\,\B(D/m+1)). \qedhere
\end{eqnarray*}
\end{proof}

\begin{lemma} \mylabel{lem:runprob} 
The total cost of all calls to \texttt{SmithMassager} in
Algorithm~\texttt{HermiteBasis} is bounded by \\
$O(m^{\omega}\, \B((\log \det S)/m + \log m)(\log m)^2 \log(1/\epsilon))$ bit
operations.  This result assumes that $m$ is a power of two.
\end{lemma} 
\begin{proof} 
Consider the top level of the recursion.  
The two calls to \texttt{SmithMassager}$[\epsilon/4]$ are with arguments
$T$ and $K$.  Both $\det T$ and $\det K$ are divisors of $\det S$.
By Theorem~\ref{thm:costSM},
the calls to \texttt{SmithMassager}$[\epsilon/4]$ have cost 
bounded by that stated in the lemma.  At level $i$ of the recursion, 
the calls to \texttt{SmithMassager} should return \textsc{fail}
with probability at most $\epsilon/4^{i+1}$, which increases
the $\log(1/\epsilon)$ factor to $(i+ \log(1/\epsilon))$.

Similar to the derivation of the bound~(\ref{eq:bsecond}) in the
proof of Lemma~\ref{lem:rundet},
the total cost of the calls to \texttt{SmithMassager} at level $i$ is bounded by
\begin{equation} 
\mylabel{eq:ccost}
m^{\omega}\, \B((\log \det S)/m + \log \bar{m})(\log \bar{m})^2 \times
\frac{i+\log(1/\epsilon)}{(2^{\epsilon_0})^i},
\end{equation} 
where $\bar{m} = m/2^i$.
To arrive at~(\ref{eq:ccost}) it suffices to replace $n$ with $m$ in~(\ref{eq:bsecond}) and to multiply by the missing
factor $(\log \bar{m})^2 (i + \log(1/\epsilon))$. 
The result now follows by 
substituting the global upper bound $\log m$ for $\log \bar{m}$
into~(\ref{eq:ccost}) and then noting that $$
\sum_{i\geq 0} \frac{i+\log(1/\epsilon)}{(2^{\epsilon_0})^i} \in O(\log(1/\epsilon)). \qedhere
$$
\end{proof}

\begin{lemma} \mylabel{lem:ext} 
If $(S,F,k,m)$ is a valid calling sequence to
Algorithm~\texttt{HermiteBasis}, then
the calling sequence 
$(\diag(I_t,S),\diag(\boldz_{t}, F),k,m+t)$ is also valid, 
$t \in \Z_{\geq 0}$.
Moreover, the Hermite basis of $\Rel(S,F)$
is the trailing submatrix of the
Hermite basis of $\Rel(\diag(I_t,S),\diag(\boldz_{t}, F))$.
\end{lemma}
\begin{proof} Let $H$ be the Hermite
basis of $\Rel(S,F)$.  The result follows from Lemma~\ref{rem:alt} by noting that 
$$
\setlength{\arraycolsep}{.7\arraycolsep} \renewcommand{\arraystretch}{.9}
\left [ \begin{array}{cc|cc}  
I_t & & \\
 & S & \\\hline 
\boldz_{t} &   & I_t & \\
& F & &  I_n \end{array} \right ] 
\mbox{~~has Hermite form~~}
\left [ \begin{array}{cc|cc}  
I_t & & \\
 & I_m & & \ast \\\hline
       &  & I_t & \\
 &   & &  H \end{array} \right ]  \qedhere
$$
\end{proof}

\paragraph{Proof of Theorem~\ref{thm:hbasrun}}
By working with the input $[\epsilon](\diag(I_t,S),\diag(\boldz_{t},F),k,m+t)$ of
Lemma~\ref{lem:ext}, where $t$ is at most $m-1$, we may assume
without loss of generality that $m$ is a power of 2.  The result
now follows as a corollary of Theorem~\ref{thm:hbcorr} (correctness),
Theorem~\ref{thm:hbprob} (probability), and Lemmas~\ref{lem:rundet} and~\ref{lem:runprob}
(running time). \qed

\section{Hermite basis of an arbitrary relations lattice}\mylabel{sec:general}

In this section we show how to compute the Hermite basis $H$ of
$\Rel(M,G)$ with arbitrary inputs.  Lemma~\ref{lem:tostandard} shows
how to construct new inputs that satisfy the preconditions of
Algorithm~\texttt{HermiteBasis}.  Theorem~\ref{thm:general} then
gives overall cost estimates for computing $H$.

\begin{lemma} \mylabel{lem:tostandard} Let $M \in \Z^{\ell \times
m}$ have full column rank and $G \in \Z^{n \times m}$.  There exists
a randomized algorithm
$ (S,F) \leftarrow \texttt{ToSmithCoprime}[\epsilon](M,G) $
that computes matrices $S$ and $F$ such that $\Rel(S,F)=\Rel(M,G)$,
with $S$ a nonsingular Smith form, $F$ reduced column-modulo $S$,
and $S$ and $F$ coprime.  The cost of the algorithm is
$O((\ell + n) m^{\omega-1} \,\B(D/m + \log m) (\log m)^2\log(1/\epsilon))$
bit operations, where 
$D$ is a bound for the sum of the bitlengths of
the columns of $M$ and $G$. 
The algorithm either returns a correct output, or reports \textsc{fail} with
probability at most $\epsilon$.
\end{lemma}
\begin{proof}
We will perform a sequence of six transformations from the input $(M,G)$ to $(S,F)$:
$$ (M,G)
\stackrel{1}{\rightarrow} \left (\left [ \begin{array}{c} M_1 \\\hline M_2 \end{array} \right ], G \right )
\stackrel{2}{\rightarrow} \left (\left [ \begin{array}{c} S_1 \\\hline M_3 \end{array} \right ], G_1\right)
\stackrel{3}{\rightarrow} (T_1, G_1)
\stackrel{4}{\rightarrow} (S_2, G_2) 
\stackrel{5}{\rightarrow} (K, C)
\stackrel{6}{\rightarrow} (S, F).
$$
Considering transformations~1 to~6 in turn we have: $M_1$ is a
nonsingular submatrix of $M$, $S_1$ is the Smith form of $M_1$,
$T_1$ is the Hermite basis of the previous modulus, $S_2$ is the
Smith form of $T_1$, $(K,C)$ are coprime inputs, and $S$ is the
Smith form of $K$.  We will show how to do each transformation
within the allotted time such that the transformed
inputs generate the same relations lattice.

Before detailing the transformations, we explain how the cost
estimate for each depends on a common upper bound $\bar{D} := D +
(m/2)\log_2 m$ that captures the bitlength of the inputs.   
By assumption, $D$, and thus also $\bar{D}$, 
bounds the sum of the bitlengths of the columns of
$M$ and $G$.  By Hadamard's bound, the base~2 logarithm of the
absolute value of any nonzero minor of $M$ is bounded by $\bar{D}$.
In particular, $\log_2 |\det M_1| \leq \bar{D}$.  The bitlength of
subsequent inputs is then bounded by $\bar{D}$ by noting that $\det S_1 = |\det
M_1|$, $(\det T_1) \mid (\det S_1)$, $\det S_2 = \det T_2$, and
$(\det K) \mid (\det S_2)$.  Moreover, $M_3$ and $G_1$ will be
reduced column-modulo $S_1$, $G_2$ will be reduced column-modulo
$S_2$, and $C$ will be reduced column-modulo the diagonals of $K$.

Transformation~1 replaces $M$ with $PM$, where $P$ is a permutation
matrix such that $PM$ has principal $m \times m$ submatrix $M_1$
nonsingular.  Correctness follows from
Lemma~\ref{lem:denequiv}.\ref{fact3}.  A suitable $P$ can be found
in the allotted time, with a chance of failure bounded by $\epsilon/4$,
by randomly choosing primes $p \in \Theta(\bar{D})$ and computing
a rank revealing matrix
decomposition~\cite{IbarraMoranHui,JeannerodPernetStorjohann13,
DumasPernetSultan17} of $A$ modulo $p$.  See, for example, \cite[Proof
of Theorem~1]{SaundersWan04}.

For the remaining transformations, the second matrix in each pair
has row dimension $n$, and $M_2$ and $M_3$ have row dimension
$\ell-m$; the cost estimate in the theorem uses the common upper bound $\ell + n$ 
for these row dimensions.  
All remaining matrices $M_1$, $S_1$,
$T_1$, $S_2$ and  $K$ are square of dimension $m$.

Transformations~2 to~6 are accomplished as follows. For each
transformation we cite the lemma that shows correctness.  For each
step, the results showing that the computation can be done in the
allotted time is cited to the right.  As noted above, a valid argument
for the function $\M$ or $\B$ in the cost estimates given by the results
cited to the right is $\bar{D}/m + \log m \in O(D/m + \log m)$.
\begin{center}

\begin{minipage}{15cm}
$\stackrel{2}{\rightarrow}$:  Lemma~\ref{thm:diag} \\
\mbox{~~~~} $(S_1,V_1)  :=  \texttt{SmithMassager}[\epsilon/4](M_1)$ \hfill \emph{\# Theorem~\ref{thm:costSM}}\\
\mbox{~~~~} $M_3  :=  \colmod(M_2V_1,S_1)$ \hfill\emph{\# Corollary~\ref{cor:partial1}}\\
\mbox{~~~~} $G_1  :=  \colmod(GV_1,S_1)$ \hfill \emph{\# Corollary~\ref{cor:partial1}}\\
\end{minipage} 

\begin{minipage}{15cm}
$\stackrel{3}{\rightarrow}$: Lemma~\ref{lem:compress}.\ref{compp1}\\
\mbox{~~~~} $T_1 := $ the Hermite basis of ${\mathcal L}(S_1) + {\mathcal L}(M_3)$
                 \hfill \emph{\# Theorem~\ref{thm:hermbas}.\ref{thm:hb1}}  \\
\end{minipage}

\begin{minipage}{15cm}
$\stackrel{4}{\rightarrow}$:  Lemma~\ref{lem:diag}\\
\mbox{~~~~} $(S_2,V_2)  :=  \texttt{SmithMassager}[\epsilon/4](T_1)$ \hfill \emph{\# Theorem~\ref{thm:costSM}}\\
\mbox{~~~~} $G_2  :=  \colmod(G_1V_2,S_2)$ \hfill \emph{\# Theorem~\ref{thm:partial}}\\
\end{minipage}

\begin{minipage}{15cm}
$\stackrel{5}{\rightarrow}$:  Lemma~\ref{thm:bas}\\
\mbox{~~~~} $T_2 := $ the Hermite basis of ${\mathcal L}(S_2) + {\mathcal L}(G_2)$ 
                  \hfill \emph{\# Theorem~\ref{thm:hermbas}.\ref{thm:hb1}} \\
\mbox{~~~~} $\left [\begin{array}{ccc} T_2 & & \ast \\
  & I & C\\
  & &  K \end{array}\right ] : = $ the  Hermite basis of $
\left [ \begin{array}{ccc} T_2 & & I \\
G_2 & I &  \\
S_2 & &  \end{array}\right ]$  \hfill \emph{\# Theorem~\ref{thm:hermbas}.\ref{thm:hb2}}  \\
\end{minipage}

\begin{minipage}{15cm}
$\stackrel{6}{\rightarrow}$:  Lemma~\ref{lem:diag}\\
\mbox{~~~~} $(S,V_3)  :=  \texttt{SmithMassager}[\epsilon/4](K)$ \hfill \emph{\# Theorem~\ref{thm:costSM}}\\
\mbox{~~~~} $F  :=  \colmod(CV_3,S)$ \hfill \emph{\# Theorem~\ref{thm:partial}}
\end{minipage}
\end{center}

Including step~1, the above recipe has four Las Vegas steps, each with probability
of failure bounded by $\epsilon/4$. The overall probability of
failure is thus bounded by $\epsilon$, as required.
\end{proof}

\begin{theorem}  \mylabel{thm:general}
Let $M \in \Z^{\ell \times m}$ have full column rank and $G \in \Z^{n \times m}$. 
There exists a Las Vegas randomized algorithm that computes the Hermite basis $H$ of a $\Rel(M,G)$ in
\begin{equation} \mylabel{eq:costx1}
O((\ell+n)m^{\omega-1}\, \B(D/m + \log m) (\log m)^2)
\end{equation}
plus 
\begin{equation} \mylabel{eq:costx2}
 O(n^{\omega}\, \B(D/n  + \log n)(\log n)^2)
\end{equation}
bit operations, where 
$D$ is a bound for the sum of the bitlengths of the column of $M$ and $G$.

Moreover, if
$n \in O(m)$, or
the Hermite basis $H$ of $\Rel(M,G)$ is $(k,\bar{m})$ for a known $(k,\bar{m})$ with $\bar{m} \in O(m)$,
then the running time is~(\ref{eq:costx1}), without requiring the extra term~(\ref{eq:costx2}).
\end{theorem}
\begin{proof}
The algorithm for the general case has two steps:
\begin{center}
\begin{minipage}{8cm}
$(S,G) := \texttt{ToSmithCoprime}[1/4](M,G)$\\
$H := \texttt{HermiteBasis}[1/4](S,G,0,n)$
\end{minipage}
\end{center}
By Lemma~\ref{lem:tostandard}, $(S,G)$ can be computed in
time~(\ref{eq:costx1}).  From the proof of Lemma~\ref{lem:tostandard}
we have that $\log_2 \det S \leq D + (m/2)\log_2 m$.  Using this
bound for $\log \det S$, Theorem~\ref{thm:hbasrun} states that $H$
can be computed in time~(\ref{eq:costx2}).  Since both steps are
Las Vegas with each having  a probability of failure  bounded by
$1/4$, the overall probability of failure is bounded by $1/2$.

Now consider the case where $H$ is index $(k,\bar{m})$ for a known
$(k,\bar{m})$ with $\bar{m} \in O(m)$.
Then we modify the second step to be
\begin{center}
\begin{minipage}{8cm}
$H := \texttt{HermiteBasis}[1/4](S,G,k,\bar{m})$.  
\end{minipage}
\end{center}
To capture the 
assumption that $\bar{m} \in O(m)$, let $\bar{m}=cm$ for some constant
$c$.  Then by Theorem~\ref{thm:hbasrun} the call the \texttt{HermiteBasis}
has cost $O(n(cm)^{\omega-1} \B(D/(cm) + \log cm)(\log cm)^2)$, which is
bounded by~(\ref{eq:costx1}) when $c$ is a constant.

Finally, if $n \in O(m)$ then we know that $H$ is index $(0,\bar{m})$
for $\bar{m}=n \in O(m)$, so this case has already been handled.
\end{proof}

\section{Applications} \mylabel{sec:examples}

In this section we give five examples of the use of our algorithm
to compute the Hermite basis of a relations lattice.  The problems
we consider take as input one or more integer matrices.  In terms
of these input matrices we define $M$ and $F$ such that $\Rel(M,F)$
is the lattice we seek.  We then use Theorem~\ref{thm:general} to
bound the cost of computing the Hermite basis of $\Rel(M,F)$.

\subsection{Hermite basis of a full column rank matrix}

The Hermite basis of a full column rank $A$ is the Hermite basis
of $\Rel(A,I)$.  Using Theorem~\ref{thm:general} to compute the
Hermite basis of $\Rel(A,I)$ gives the following.

\begin{theorem} \mylabel{thm:EXhermite} Let $A \in \Z^{n \times m}$
have full column rank.  There exists a Las Vegas randomized
algorithm that computes the Hermite form of $A$ in $O(nm^{\omega-1}\,\B(D/m
+ \log m)(\log m)^2))$ bit operations, where $D$ is a bound for the
sum of the bitlengths of the columns of $A$.
\end{theorem}

\subsection{Remainder modulo a Hermite form}
As done in~\cite{NeigerVu17} for polynomial matrices, we can compute
the remainder of a matrix with respect to a Hermite basis, but now
over $\Z$.  Recall that
the remainder of $F \in \Z^{n \times m}$ with respect to
a Hermite basis $T \in \Z^{m \times m}$ is 
the matrix $\bar{F}$ of
Lemma~\ref{lem:compress}.\ref{compp2}.
Here, we observe that 
\begin{equation*}
\Rel \left ( T, \left [ \begin{array}{c} -F \\\hline I \end{array} \right ] \right )
\mbox{~~has Hermite basis~~}
H = \left [\begin{array}{c|c} I & \bar{F} \\\hline & T \end{array} \right ].
\end{equation*}
We obtain the following result using Theorem~\ref{thm:general} to
compute $H$, using the fact that $H$ is index $(n,m)$.

\begin{theorem} Let $T \in \Z^{m \times m}$ be a Hermite basis and
let $F \in \Z^{n \times m}$. There exists a Las Vegas randomized
algorithm that computes the remainder of $F$ with respect to $T$
in $$O((n+m)m^{\omega-1}\, \B(D/m + \log m) (\log m)^2)$$ bit operations,
where $D$ is bound for the sum of the bitlengths of the columns of $T$ and $F$.
\end{theorem}

\subsection{Hermite form of a product of matrices}

Suppose $A$ and $B$ are nonsingular.  An obvious way to compute the
Hermite form $H$ of the product $AB$ is to compute $C = AB$
explicitly, and then compute the Hermite form of $C$.  But the total
size of $C$ can be large.  For example, consider the $(2n+1)
\times (2n+1)$ matrices
\begin{equation*} 
A = \left [ \begin{array}{ccc} I_n & 2^nv & \\
 &  2^n+1 & \\
& & I_n \end{array} \right ] \hbox{~~and~~}
B = \left [ \begin{array}{ccc} I_n &   & \\
 & 1  & w \\
 &  & 2I_n \end{array} \right ],
\end{equation*}
where $v \in \Z^{n \times 1}$ and $w \in \Z^{1 \times n}$ are the
column and row vector of all ones, respectively.  Both $A$ and $B$
are in Hermite form, and the sum of the bitlengths of all entries
in $A$ and $B$, as well as  the Hermite form $H$ of $AB$, is $O(n^2)$.
But $C=AB$ contains the $n \times n$ submatrix $2^nvw$ filled with
$2^n$, and so the sum of the bitlength of all entries in $C$ is
$\Omega(n^3)$.  Instead of computing $C$ explicitly, we can use the
fact that the Hermite basis of $AB$ is equal to the Hermite basis
of $\Rel(AB,I)$.  Up to transformations allowed by Lemma~\ref{lem:denequiv}
we have
\begin{equation} \Rel(AB,I) 
  = \Rel \Biggl ( \left [ \begin{array}{c|c} A & \\ I  & B \end{array} \right ]  ,
 \left [ \begin{array}{c|c} \boldz & I \end{array} \right] \Biggr ).  \mylabel{eq:s2}
\end{equation}
Note that identity~(\ref{eq:s2}) applies also to rectangular matrices $A$ and
$B$.  Using Theorem~\ref{thm:general} to compute the Hermite basis
of the right hand side of~(\ref{eq:s2}) gives the following.
\begin{theorem} Let $A \in \Z^{n \times m}$ and $B \in \Z^{m \times
\ast}$ be such that $AB$ has full column rank.  There exists a Las Vegas
randomized algorithm that computes the Hermite basis of $AB$ in
$$O((n+m)m^{\omega-1}\,\B(D/m + \log m)(\log m)^2)$$ bit operations, where
$D$ is a bound for the sum of the bitlengths of the columns of $A$
and $B$.
\end{theorem}

\subsection{Intersection of bases}

Suppose $A,B \in \Z^{n \times m}$ have full column rank.  Then
\begin{equation} \mylabel{eq:inpy}
{\mathcal L}(A) \cap {\mathcal L}(B) = 
\Rel \Biggl ( \left [ \begin{array}{c|c} A & \\ & B \end{array} \right ],
\left [ \begin{array}{c|c} I & I \end{array} \right ] \Biggr ).
\end{equation}
Using Theorem~\ref{thm:general} to compute the Hermite basis of the
right hand side of~(\ref{eq:inpy}) gives the following.

\begin{theorem} Let $A,B \in \Z^{n \times m}$ have full column
rank.  There exists a Las Vegas randomized algorithm that computes
the Hermite basis of ${\mathcal L}(A) \cap {\mathcal L}(B)$ in
$O(nm^{\omega-1}\,\B(D/m + \log m)(\log m)^2)$ bit operations, where
$D$ is a bound for the sum of the bitlengths of the columns
of $A$ and $B$.
\end{theorem}

\subsection{Multivariable Chinese remainder problem}

The Chinese remainder problem asks to 
find an integer $x$ satisfying 
\begin{equation} \mylabel{eq:crt}
\setlength{\arraycolsep}{.7\arraycolsep} \renewcommand{\arraystretch}{.7}
x \left [ \begin{array}{cccc}
1 &  1 &  \cdots & 1 \end{array} \right ] = 
\stackrel{\textstyle \vec{b}}{
\left [ \begin{array}{cccc}
b_1 &  b_2 & \cdots & b_n  \end{array} \right ]} \bmod 
\stackrel{\textstyle M}{ \left [ \begin{array}{cccc} d_1 & \\
 & d_2 &  \\
& & \ddots & \\
& & & d_n \end{array} \right ]}\end{equation}
for a given vector $\vec{b} \in \Z^{1 \times n}$ 
of remainders and a diagonal matrix $M\in \Z_{\geq 0}^{n \times n}$
of nonzero moduli.
In the classical setting one typically makes the assumption that the
moduli $d_i$ are pairwise relatively prime --- in which case there is a
 unique solution for $x$ in the range $[0,\det M)$ --- but the case of arbitrary moduli has
also been considered~\cite{Ore52}.
Another natural generalization of~(\ref{eq:crt}) takes as input
an $A \in \Z^{n \times n}$ and asks for a vector $\vec{x} \in \Z^{1 \times n}$ (if it exists) to
the linear diophantine system
$\vec{x} A = \vec{b} \bmod M.$
Knill~\cite{Knill12} and Sury~\cite{Sury15} give additional background
and history, referring to this linear diophantine system 
as the {\em multivariable Chinese remainder problem}.

Here we observe that 
\begin{equation}\mylabel{crt3}
\renewcommand{\arraystretch}{1.2}
\Rel \Biggl ( M, \left [\begin{array}{c} -\vec{b} \\\hline A \end{array} \right ] 
\Biggr )
\mbox{~~has Hermite basis~~}
\left [ \begin{array}{c|c} h & \vec{x}_p \\\hline
 & \bar{H} \end{array} \right ],
\end{equation}
where $\vec{x}_p \in \Z^{1 \times n}$  is a particular solution to the scaled system
$\vec{x}A  = h\vec{b} \bmod M$,
and the scaling factor $h\in \Z_{>0}$ is minimal to achieve
consistency. 
Using Theorem~\ref{thm:general} to compute the Hermite basis 
of~(\ref{crt3}) gives the following.  

\begin{theorem} \mylabel{thm:mcrt}  
Let $M\in \Z_{\geq 0}^{n \times n}$ be nonsingular and diagonal, 
$A \in \Z^{n \times n}$ and $\vec{b} \in \Z^{1 \times n}$. 
There exists a Las Vegas randomized algorithm that
computes:
\begin{itemize}
\item[(i)]  The minimal $h \in \Z_{>0}$ such that $\vec{x} A = h\vec{b} \bmod M$ has a solution for $\vec{x} \in \Z^{1 \times n}$.
\item[(ii)]  A particular solution vector $\vec{x}_p \in \Z^{1 \times n}$ such that $\vec{x}_p A = h\vec{b} \bmod M$.
\item[(iii)]  A Hermite basis $\bar{H} \in \Z^{n \times n}$ such that 
$\{ \vec{x} \, \mid \, \vec{x} A = h\vec{b} \bmod M \} = \{ \vec{x}_p + v \bar{H} \, \mid \, v \in \Z^{1 \times n} \}$.
\end{itemize}
If $A$ and $\vec{b}$ are reduced  column-modulo $M$, then
the cost of the algorithm is 
$$O(n^{\omega}\,\B((\log \det M)/n + \log n)(\log n)^2)$$ 
bit operations.
\end{theorem}

\begin{remark}
While computing a basis for a lattice of integer relations has
not had considerable attention to date, the same is not true for
the case of polynomials. In this case the CRT corresponds to polynomial
interpolation while modules of polynomial relations define such
well known problems as rational interpolation, multi-point Pad\'e
approximation and more generally M-Pad\'e approximation along with
their matrix generalizations~\cite{BeckermannLabahn97,NeigerVu17}.
\end{remark}

\section{Conclusion and topics for future research}\mylabel{sec:conc}

Let $M \in \Z^{n \times n}$ be a nonsingular integer matrix and $d=
\log n  + \log ||M||$.  In this paper we have given a Las Vegas
randomized algorithm to compute the Hermite form of $M$ using
\begin{equation} \mylabel{eq:ccost1} O(n^{\omega} (\log n)^2 \,
d^{1+\epsilon}) 
\end{equation} 
bit operations, where $\omega$ is the exponent of matrix multiplication
and the $+\epsilon$ captures the cost of integer arithmetic, which
can be pseudo-linear.  We make the common assumptions that $\omega > 2$
and $\epsilon < \omega-2$.

Our approach is to solve a more general problem, that of computing
the basis in Hermite form of an integer relations lattice. To compute
the Hermite form $H$ of $M$, we start with a Smith massager $(S,F)$
for $M$ that can be computed in cost~(\ref{eq:ccost1}) using an
existing Las Vegas randomized algorithm.  We then give a divide and
conquer algorithm to compute the Hermite basis of the relations
lattice $\Rel(S,F)$, which is equal to $H$, also in Las Vegas
cost~(\ref{eq:ccost1}).  The divide and conquer algorithm itself
is randomized because it uses Smith massager computations during
the construction of the inputs for the recursive calls.

In terms of complexity,
the main open problem concerns the need for randomization:
find (deterministic) algorithms to compute the Hermite
form $H$ of $M$ and the Smith form $S$ of $M$ using
\begin{equation} \mylabel{eq:ccost2}
(n^{\omega}d)^{1+o(1)}
\end{equation} bit operations.
Regarding the Smith form, on the one hand, the existence of an algorithm for the more general
problem of computing a Smith massager $(S,F)$ of $M$ in
cost~(\ref{eq:ccost2}) would answer positively the question about
computing $H$.  On the other hand, finding just the largest invariant
factor of $S$ in cost~(\ref{eq:ccost2}) is open.


A natural analogue to the problems we consider is computing Hermite
and Smith forms of matrices over the ring of univariate polynomials
$\K[x]$, $\K$ a field.  Over $\K[x]$ the cost is 
the number of required field operations from $\K$ and 
$d$ is a bound for the degrees of entries in the input matrix.
Given a nonsingular $M \in \K[x]^{n \times n}$, 
deterministic algorithms with cost~(\ref{eq:ccost2}) are known for computing
the Hermite form~\cite{LabahnNeigerZhou17} and 
for computing 
the largest invariant factor~\cite{ZhouLabahnStorjohann14}.
However, the problem of computing the Smith form in cost~(\ref{eq:ccost2}) remains open.

Our algorithm constructively reduces the problem of computing the
Hermite form to that of  multiplying integer matrices, and is thus  well
suited to take advantage of hardware and software support
both for matrix multiplication and fast integer arithmetic.  At
present we have a preliminary C implementation for the computation of a Smith
massager \cite{Ziwen}, one of the core subroutines needed.  
It remains to implement the fast Hermite and Howell constructions along with the
specialized column-modulo matrix multiplication operations.  
Such implementations require developing techniques to bridge
the gap between an algorithm with good asymptotic complexity and
one having practical performance.

\bibliographystyle{plainnat}

\newcommand{\SortNoop}[1]{}

\end{document}